\documentclass[letterpaper]{article}

\usepackage{aaai} 
\nocopyright
\usepackage{etex}
\usepackage{lmodern}
\usepackage[T1]{fontenc}

\usepackage{latexsym}

\usepackage{mathtools}

\usepackage{amsopn}

\PassOptionsToPackage{table}{xcolor}
\usepackage{graphics}
\usepackage{tikz}
\usepackage[all]{xy}
\usepackage{url}
\usepackage{enumerate}

\usepackage{booktabs}

\usepackage{amsmath}
\usepackage{amssymb}
\usepackage{thmtools,thm-restate}
\usepackage[all]{xy}
\usepackage{xspace}
\usepackage{tikz}
\usepackage{lmodern}
\usepackage[T1]{fontenc}
\usepackage{graphicx}
\usetikzlibrary{arrows,backgrounds,positioning}
\usepackage{colortbl}

\usepackage[ruled,vlined,linesnumbered,lined]{algorithm2e}

\newcommand{\game}{\mathcal{G}_\phi}
\newcommand{\NE}{Nash equilibrium\xspace}
\newcommand{\NEs}{Nash equilibria\xspace}

\newcommand{\ene}{$\exists${\sc \small NE}\xspace}
\newcommand{\ane}{$\forall${\sc \small NE}\xspace}

\newcommand{\NP}{{NP}\xspace}
\newcommand{\coNP}{{coNP}\xspace}
\newcommand{\calO}{\mathcal{O}}
\newcommand{\neG}{\text{NE}(\mathcal{G})}
\newcommand{\SP}{\mathit{SP}}
\def\single{singleton\xspace}

 \newtheorem{theorem}{Theorem}
 
 \newtheorem{lemma}{Lemma}
 \newtheorem{example}{Example}

\newcounter{symbol}
\newcommand{\indexsyma}[1]%
{\stepcounter{symbol}\index{zzz1 \thesymbol @\protect#1}}
\newcommand{\indexsymb}[1]%
{\stepcounter{symbol}\index{zzz2 \thesymbol @\protect#1}}
\newcommand{\indexsymc}[1]%
{\stepcounter{symbol}\index{zzz3 \thesymbol @\protect#1}}
\newcommand{\indexsymd}[1]%
{\stepcounter{symbol}\index{zzz4 \thesymbol @\protect#1}}
\newcommand{\indexsyme}[1]%
{\stepcounter{symbol}\index{zzz5 \thesymbol @\protect#1}}

\newcommand{\bfe}[1]{\begin{bfseries}\emph{#1}\end{bfseries}\index{#1}}

\newcommand{\myra}{\mbox{$\:\rightarrow\:$}}

\newcommand{\LL}{\mbox{$\ldots$}}

\newcommand{\szkew}[1]{\relax \setbox0=\hbox{\kern -24pt $\displaystyle#1$\kern 0pt }%
\box0}
{\catcode`\@=11 \global\let\ifjusthvtest@=\iffalse}
\newcounter{oldmycaption}

\newcommand{\wgt}[2]{{w_{#1 \to #2}}}

\newcommand{\strprofile}{s}

\usepackage{wrapfig}
\newcommand{\citey}[1]{(\citeyear{#1})\xspace}

\newcommand{\qed}{\hfill$\Box$}

\title{Constrained Pure Nash Equilibria in Polymatrix Games}
\author{Sunil Simon\thanks{Supported by the Research-I Foundation, IIT Kanpur and the Liverpool-India fellowship,
University of Liverpool.}\\ IIT Kanpur\\Kanpur, India\\
\And Dominik Wojtczak\thanks{Supported by EPSRC grant
EP/M027651/1.}\\University of Liverpool\\
Liverpool, U.K. \\
}

\begin{document}
\maketitle

\begin{abstract}
We study the problem of checking for the existence of constrained pure Nash
equilibria in a subclass of polymatrix games defined on weighted
directed graphs. The payoff of a player is defined as the sum of
nonnegative rational weights on incoming edges from players who picked the same
strategy augmented by a fixed integer bonus for picking a given
strategy.  These games capture the idea of coordination within a local
neighbourhood in the absence of globally common strategies.  We study
the decision problem of checking whether a given set of
strategy choices for a subset of the players is consistent with some
pure Nash equilibrium or, alternatively, with all pure Nash
equilibria. We identify the most natural
tractable cases and show \NP or \coNP-completness of these problems
already for unweighted DAGs.
\end{abstract}
\date{}

\section{Introduction}
\label{sec:intro}

Identifying subclasses of games where equilibria is tractable is an
important problem in algorithmic analysis of multiplayer games. Pure
Nash equilibria (NEs) may not exist in games and checking whether a
game has a pure NE is in general a hard problem.  Even for subclasses
of games in which a pure NE is guaranteed to exists (for instance,
potential games) computing one remains PLS-hard
\cite{fabrikant_complexity_2004}.  Although, Nash's theorem guarantees
the existence of mixed strategy NE in all finite games, computing one
is still a hard problem. %
Therefore, identifying restricted classes of
games where equilibrium computation is tractable and also precisely
identifying the borderline between tractability and hardness in such
restricted classes is of obvious interest. In this paper, we study the
borderline of tractability in a natural subclass of games where the
utilities of players are restricted to be pairwise separable. These
are called \textit{polymatrix games} \cite{Jan68} and they form an
abstract model that is useful to analyse strategic behaviour of
players in games formed via pairwise interactions.
In polymatrix games, the payoff for each player is the sum of the
payoffs he gets from individual two player games he plays against
every other player. 
Polymatrix games are well-studied in the literature and include game classes
with good computational properties like the two-player zero-sum
games. They also have applications in areas such as artificial neural
networks \cite{MZ91} and machine learning \cite{EP12}.

In terms of tractability, the restriction to pairwise interactions does
not immediately ensure the existence of efficient
algorithms. Computing a mixed strategy Nash equilibrium remains
PPAD-complete \cite{CD11} and checking for the existence of a pure NE 
is NP-complete in general. This motivates the need to
further analyse the type of pairwise interactions that would ensure tractability.
In this paper, we argue that another important factor which influences
tractability is the structure of the underlying interaction graph and
presence of individual preferences (that we call {\em bonuses}).

The main restriction that we impose on polymatrix games is that each
pairwise interaction form a coordination game. Henceforth, we will
refer to these games simply as \textit{coordination games on
  graphs}. Coordination games are often used in game theory to model
situations where players attain maximum payoff when they agree on a
common strategy. The game model that we study, extends coordination
games to the network setting where payoffs need not always be
symmetric and players coordinate within a certain local
neighbourhood. The neighbourhood structure is specified by a finite
\textit{directed} graph whose nodes correspond to the players. Each
player chooses a colour from a set of available colours.  The payoff
of a player is the sum of weights on the edges from players who choose
the same colour and a fixed bonus for picking that particular colour.
This game model is closely related to various well-studied classes of
games. For instance, coordination games on graphs are
\textit{graphical games} \cite{KLS01} and they are also related to
\textit{hedonic games} \cite{DG80,BJ02}. In hedonic games, the payoff
of each player depends solely on the set of players that selected the
same strategy.
The coalition formation property inherent to coordination games on
graphs make the game model relevant to \textit{cluster analysis}. 
The problem of clustering has been studied from a game theoretic
perspective for instance in \cite{feldman2012hedonic,PB14}. 
Feldman and Friedler \citey{FF15} introduced a framework for the analysis of clustering games on networks
where the underlying coalition formation graph is
undirected and, 
as a result, a potential game. 
Hoefer \citey{Hoefer2007} also
studied clustering games that are polymatrix games based on undirected
graphs where each player has the same set of strategies.
These games are also potential games.
Coordination games on graphs constitute a game model which can be
useful for analysing the adoption of a product or service 
within a network of agents 
interacting with each other in their local neighbourhoods.
For example, consider the selection of a mobile
phone operator. %
The interaction between users can be represented by
a coordination game where the weight of the edge
from $i$ to $j$ represents the total cost of calls from $j$ to $i$.
Also, the bonus function can represent individual preferences of users over the providers.
Now suppose that mobile network operators allow free calls among its
users. Then each mobile phone user faces a strategic choice of
picking an operator
that maximises his
cost savings
or, in the case of unweighted graphs, 
maximises the number
of people he can call for free. 
If players are allowed to freely switch their operator based on their
friends' choices, then the stable market states correspond to pure
\NEs in this game.  One can observe similar interactions in
peer-to-peer networks, social networks and photo sharing platforms.

A similar game model based on \textit{undirected} graphs was
introduced in \cite{ARSS14} and further studied in \cite{RS15}.  The
transition from undirected to directed graphs drastically changes the
status of the games. For instance, in the case of undirected graphs,
coordination games are potential games whereas in the directed case,
Nash equilibria may not even exist. Moreover, the problem of
determining the existence of pure NEs is NP-complete for
coordination games on directed graphs \cite{ASW15}. 
However, pure NE always exists for several natural classes of graphs \cite{SW16}.

However, in many practical situations, finding just one pure \NE may not be
enough.  In fact, there can be exponentially many \NEs, each with a
different payoff to each player (see Example \ref{ex:many-NEs}).
Ideally, we would like to ask for the existence of a \NE satisfying
some given constraints.  In this paper, we focus on checking whether a
partial strategy profile (i.e.  strategy choices for a subset of the
players) is consistent with some pure \NE or, alternatively, with all
pure Nash equilibria.  We will refer to these as \ene and \ane
decision problem, respectively.  We identify the most natural
tractable cases and show \NP or \coNP-completness of these problems
already for unweighted DAGs.

\smallskip\noindent{\bf Related work.} 
The complexity of checking for the existence of pure Nash equilibria
in a game crucially depends on the representation of the game. \textit{Normal form} representation can be exponential
in the number of players whereas graphical
games and polymatrix games provide a more concise representation of
strategic form games.
While checking for the existence of pure Nash equilibria can be solved
in {\sc LogSpace} for games in normal form, it is \NP-complete for
graphical games even when the payoff of each player depends only on
the strategy choices of at most three other players
\cite{gottlob_pure_2005}.  On the other hand, it is solvable in
polynomial time for graphical games whose dependency graph has a
bounded treewidth \cite{gottlob_pure_2005} or when each player has
only two possible strategies \cite{thomas_pure_2015}. For polymatrix
games, checking for the existence of a pure \NE is \NP-complete even
when all its individual 2-player games are win-loss ones \cite{ASW15}.

Gilboa and Zelmel \citey{gilboa_nash_1989} were the first to study the
computational complexity of decision problems for mixed \NEs with
additional constraints for two player games in normal form.  For many
natural constraints the corresponding decision problems were shown to
be \NP-hard. Further hardness results were shown
in \cite{conitzer_new_2008} and \cite{bilo_complexity_2012}.
The existence of constrained pure NE can be solved in {\sc LogSpace}
for normal form games simply by checking every pure strategy
profile. For graphical games the problem is {\sc NP}-hard even without
any constraints \cite{gottlob_pure_2005}, 
but because of the special structure of our games, 
this result does not directly apply in our setting.
On the other hand, constrained pure NE can be found in
polynomial time for graphical games played on graphs with a bounded
treewidth \cite{greco_complexity_2009}.  
We are not aware of any prior work on this problem for polymatrix games.
Our paper is the first to identify several subclasses of polymatrix
games for which the existence problem of a constrainted \NE 
is tractable.

\section{Background}
\label{sec:prelim}

A \bfe{strategic game} $\mathcal{G}=(S_1, \ldots, S_n,$ $p_1, \ldots,
p_n)$ with $n > 1$ players consists of a non-empty set $S_i$ of
\bfe{strategies} and a \bfe{payoff function} $p_i : S_1 \times \cdots
\times S_n \myra \mathbb{R}$, for each player $i \in \{1,2,\ldots,n\}$. 
Let $S := S_1 \times \cdots \times S_n$ and let us call each element
$s \in S$ a \bfe{joint strategy}. Given a joint strategy $s$, we
denote by $s(i)$ the strategy of player $i$ in $s$. We abbreviate the
sequence $(s(j))_{j \neq i}$ to $s_{-i}$ and occasionally write $(s(i),
s_{-i})$ instead of $s$. 
We call a strategy $s(i)$ of player $i$ a \bfe{best
response} to a joint strategy $s_{-i}$ of his opponents if for all $
x \in S_i$, $p_i(s(i), s_{-i}) \geq p_i(x, s_{-i})$.
We do not consider mixed strategies in this paper.

Given two joint strategies $s'$ and $s$, we say
that $s'$ is a \bfe{deviation of the player $i$} from $s$ 
if $s_{-i} = s'_{-i}$ and $s(i) \neq s'(i)$.
If in addition $p_i(s') > p_i(s)$, we say
that the deviation $s'$ from $s$ is \bfe{profitable} for player $i$.
We call a joint strategy $s$ a (pure) \bfe{Nash equilibrium}
if no player can profitably deviate from $s$. 
For any given strategic game $\mathcal{G}$, let $\neG$ denote the set
of all (pure) \NEs in $\mathcal{G}$.

\label{sec:model}

We now introduce the class of games we are interested in.  Fix a
finite set of colours $M$.
A weighted directed graph $(G,w)$ is a
structure where $G=(V,E)$ is a graph without self loops over the
vertices $V=\{1,\ldots,n\}$ and $w$ is a function that associates with
each edge $e \in E$, a nonnegative rational weight $w_e \in \mathbb{Q}_{\geq 0}$.
We say that a node $j$ is a \bfe{successor} of the node $i$, and $i$
is a \bfe{predecessor} of $j$, if there is an edge $i \to j$ in $E$.
Let $N_i$ denote the set of all predecessors of node $i$ in the graph
$G$.  By a \bfe{colour assignment} we mean a function that assigns to
each node of $G$ a finite non-empty set of colours.  A \bfe{bonus} is
a function $\beta$ that to each node $i$ and a colour $c$ assigns an
integer $\beta(i,c)$.

Given a weighted graph $(G,w)$, a colour assignment $C : V \to
2^M\setminus\{{\emptyset}\}$ and a bonus function $\beta : V \times M
\to \mathbb{Z}$, a strategic game $\mathcal{G}(G,w,C,\beta)$ is
defined as follows:
\begin{itemize}
\itemsep0em 
\item the players are the nodes;
\item the set of strategies of player (node) $i$ is the set of colours
  $C(i)$; 
\item the payoff function $p_i(s) := \sum_{j \in N_i : \, s(i) = s(j)} w_{j
  \to i} + \beta(i,s(i))$.
\end{itemize}

So each node simultaneously chooses a colour and its payoff is the sum
of the weights of the edges from its neighbours that chose the same
colour augmented by a bonus to the node from choosing this colour.  We
call these games \bfe{coordination games on directed graphs}, from now
on just \bfe{coordination games}.  When the weights of all the edges
are 1, we obtain a coordination game whose underlying graph is
unweighted.  In this case, we simply drop the function $w$ from the
description of the game.
In this case the payoff function is defined by $p_i(s) :=
|\{j \in N_i \mid s_i = s_j\}| + \beta(i,s(i))$.
Similarly if all the
bonuses are $0$, we obtain a coordination game without
bonuses. Likewise, to denote this game we omit the function
$\beta$. 
Note that an edge with positive integer weight $w$ can be simulated by adding $w$ nodes and 
$2w$ unweighted edges to the game, and any positive integer bonus can be simulated similarly.
However, if all weights and bonuses are represented in binary, as we assume in this paper,
such an operation can increase the size of the graph exponentially and 
be inefficient.
\begin{figure}
\centering
\tikzstyle{agent}=[circle,draw=black!80,thick, minimum size=2em,scale=0.8]
\begin{tikzpicture}[auto,>=latex',shorten >=1pt,on grid]
\newdimen\R
\R=1.3cm
\newcommand{\llab}[1]{{\small $\{#1\}$}}
\draw (90: \R) node[agent,label=right:{\llab{a,\underline{b}}}] (1) {1};
\draw (90-120: \R) node[agent,label=right:{\llab{a,\underline{c}}}] (2) {2};
\draw (90-240: \R) node[agent,label=left:{\llab{b,\underline{c}}}] (3) {3};
\draw (30: \R) node[agent,label=right:\llab{a,\underline{b}}] (4) {4};
\draw (30-120: \R) node[agent,label=right:{\llab{a,\underline{c}}}] (5) {5};
\draw (30-240: \R) node[agent,label=left:{\llab{b,\underline{c}}}] (6) {6};
\draw (90: 1.7*\R) node[agent,label=right:{\llab{\underline{a}}}] (7) {7};
\draw (90-120: 2*\R) node[agent,label=right:{\llab{\underline{c}}}] (8) {8};
\draw (90-240: 2*\R) node[agent,label=left:{\llab{\underline{b}}}] (9) {9};
\foreach \x/\y in {1/2,2/3,3/1,1/4,4/2,2/5,5/3,3/6,6/1,7/1,8/2,9/3} {
    \draw[->] (\x) to (\y);    
}
\end{tikzpicture}
\caption{Unweighted coordination game with no NE. %
\label{fig:graph}
\vspace*{-0.5em}
}
\end{figure}
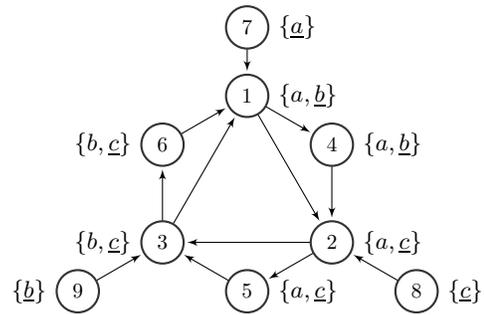
\begin{example} \label{exa:payoff}
\rm
Consider the unweighted directed graph and the colour assignment
depicted in Figure~\ref{fig:graph}.
Take the joint strategy $s$ that consists of the underlined strategies.
Then the payoffs are as follows:
{\bf 0} for the nodes 1, 7, 8, and 9; 
{\bf 1} for the nodes 2, 4, 5, and 6; 
{\bf 2} for the node 3.

Note that $s$ is not a Nash equilibrium. For example,
node 1 can profitably deviate to colour $a$.
In fact the coordination game associated with this graph does not
have a Nash equilibrium.
Note that for nodes 7, 8 and 9 the only
option is to select the unique strategy in its strategy set. The best
response for nodes 4, 5 and 6 is to always select the same strategy as
nodes 1, 2 and 3, respectively. Therefore, to show that the game does
not have a Nash equilibrium, it suffices to consider the strategies of
nodes 1, 2 and 3. We denote this by the triple
$(\strprofile_1,\strprofile_2,\strprofile_3)$. Below we list all such
joint strategies and we underline a strategy that is not a best
response to the choice of other players: $(\underline{a},a,b)$,
$(a,a,\underline{c})$, $(a,c,\underline{b})$,
$(a,\underline{c},c)$, $(b,\underline{a},b)$,
$(\underline{b},a,c)$, $(b,c,\underline{b})$ and
$(\underline{b},c,c)$.  
\qed
\end{example}

Let $Q \subseteq V$ be a nonempty subset of all the nodes of a given graph
$G$. A \bfe{query} is a function $q: Q \to M$ which satisfies the
following property: for all $i \in Q$, $q(i) \in C(i)$. 
We say that a query $q$ is \bfe{consistent} with a strategy profile $s$ iff
$q = s|_{Q}$, i.e. $q(i) = s(i)$ for all $i \in Q$.
We call a query $q:Q \to M$ \bfe{monochromatic} if for all $i, j \in
Q$, $q(i) = q(j)$ and otherwise we call the query \bfe{polychromatic}.
A query $q$ is said to be \bfe{\single} if $|Q| = 1$. Obviously
every \single query is also a monochromatic one.
In this
paper, we study the following decision questions. 

{\itshape
\smallskip\noindent
Given a graph $G=(V,E)$, weights $w$, 
colour assignment $C$, bonus function $\beta$, and query $q$.
\vspace*{-.1em}

\smallskip
\noindent {\em(\ene problem)} Is there a Nash equilibrium in $\mathcal{G}(G,w,C,\beta)$ that is consistent with $q$?

\smallskip
\noindent {\em(\ane problem)} Is every Nash equilibrium in $\mathcal{G}(G,w,C,\beta)$ consistent with $q$?
}

\smallskip
\noindent Formally, \ene problem asks if there exists $s \in \neG$ such that $q
= s|_{Q}$, while the
\ane problem asks whether for all $s \in \neG$ it is the case that $q = s|_{Q}$.
Note that \ane is not a complement of \ene. Actually, any
non-\single \ane query can be reduced to a series of
\single \ane queries $q|_{\{i\}}$ for every player $i \in Q$. 
Note that trivially \ene $\in$ \NP and \ane~$\in$~\coNP, because checking whether 
a joint strategy is a \NE and is consistent with $q$ can be done in polynomial time.

\begin{table}
\setlength{\tabcolsep}{0.05cm}
\renewcommand{\arraystretch}{1.2}
\rowcolors{2}{white}{gray!20}
\resizebox{\columnwidth}{!}{%
\begin{tabular}{ccc}
\noalign{\global\belowrulesep=0.0ex}
\toprule
\noalign{\global\aboverulesep=0.0ex}
Graph Class & \ene & \ane\\ %
\midrule
2 colours+monochromatic query & $\mathcal{O}(|G|)$ & $\mathcal{O}(|G|)$\\
2 colours+polychromatic query & NP-comp. & $\mathcal{O}(|G|)$\\
DAGs+3 colours+\single query & NP-comp. & coNP-comp.\\
simple cycles & $\mathcal{O}(|G|)$ & $\mathcal{O}(m\cdot|G|)$\\
DAGs with out-degree $\leq 1$ & $\mathcal{O}(|G|^{2.5})$ & $\mathcal{O}(|G|^{2.5})$\\
colour complete graphs no bonuses & $\mathcal{O}(nm\cdot m!)$ & $\mathcal{O}(nm\cdot m!)$\\
\noalign{\global\aboverulesep=0.0ex}
\bottomrule
\end{tabular}
}
\caption{
Summary of the results. The last two classes are unweighted; a simple reduction from the {\sc Partition} problem and its complement, shows \NP and \coNP hardness of their \ene and \ane problems, respectively, in the weighted case.}
\label{tab:results}
\end{table}

Given a directed graph $G$ and a set of nodes $K$, we denote by $G[K]$
the subgraph of $G$ induced by $K$. A (directed) graph $G=(V,E)$ is a \bfe{complete graph} if for all 
$i,j \in V$ such that $i \neq j$, we have $i \to j \in E$.
That is from every node there is an edge to every other node.
Given the set of colours $M$, we say that a directed
graph $G$ is \bfe{colour complete} (with respect to a colour
assignment $C$) if for every colour $c \in M$ each component of
$G[V_c]$ is a complete graph, where $V_c = \{i \in V \mid c \in C(i)\}$. In
particular, every complete graph is colour complete, but not vice versa
(see Figure \ref{fig:colour-complete-graph} in the appendix).

Table~\ref{tab:results} summarises our results in terms of the number
of arithmetic operations needed.  We use binary representation for all
values in $w$ and $\beta$.  The size of the input game graph is $|G|
= \mathcal{O}(nm+e)$, where $n$ is the number of nodes in a graph, $m$
is the number of colours and $e$ is the number of edges.  

Note that these graph classes can occur naturally in practice.  Graphs
with two colours can model duopoly markets and simple cycles are used
in Token ring architectures. Unweighted DAGs with out-degree $\leq 1$
can model indirect elections such as the US primaries where votes are
cast for delegates, who may have their own preferences, rather than
for presidential nominees directly. %
In this context, the \ene question answers who can become the leader
based on the list of candidates each voter realisticly considers
voting for (represented by the set of available colours) and \ane can
tell us if a given candidate wins no matter how the undecided voters
(i.e.  players with non-singleton set of available colours)
vote. Colour complete graphs can model situations where every user
benefits as the number of users increases even if they do not know
each other directly, e.g.~users joining a torrent swarm. Also, in the
context of a market with multiple products, the \ene /\ane %
questions can tell us which product can/will dominate the market in
the end.

\section{Graphs with Two or Three Colours}
\label{sec:two-colours}

We start by studying coordination games with two colours and monochromatic queries.
To fix the notation, let $G=(V,E)$ and the colour set be $M=\{0,1\}$. 
Let $q$ be a monochromatic query.
Without loss of generality, we can assume $q(i) = 0$ for all $i\in Q$, because
otherwise we can rename the colours. We show how to deal with the \ene decision problem first.

\begin{algorithm}
\caption{\label{alg:2c-ene}
Algorithm for \ene on arbitrary graphs with two colours and monochromatic queries.} 
\KwIn{A coordination game $\mathcal{G}((V,E),w,C,\beta)$ and monochromatic query $q : Q \to M$.
} 
\KwOut{YES if there exists a \NE consistent with $q$ and NO otherwise.} 

\For{$i \in V$}{
  {\bf if} {$0 \in C(i)$} {\bf then} {$s(i) = 0$} {\bf else} {$s(i) = 1$} 
}
{\bf set} $\mathcal{S}:= \{i\ |\ s(i) = 1\}$ 

\While{$\mathcal{S} \neq \emptyset$}{
   $\text{remove any element from } \mathcal{S}$ {and assign it to $i$}\\
   \For{$\{j \in V$ | $i \to j \in E\}$}{
      \If{$s(j) = 0$ {\bf and} $1 \in C(j)$ {\bf and} $p_j((1,s_{-j}))>p_j(s)$}
      {    $s(j) = 1$ \\
           $\text{add $j$ to }\mathcal{S}$
      }     
   }
}   

{\bf if} $\forall_{i\in Q}\ s(i) = 0$ {\bf return} YES {\bf else} {\bf return} NO
\end{algorithm}

\begin{restatable}{theorem}{twocolourmonoene}
\label{thm:2colour-mono-ene}
The \ene problem for coordination games with two colours and
monochromatic queries can be solved in $\calO(|G|)$ time
using Algorithm \ref{alg:2c-ene}. 
\end{restatable}

Similarly, Algorithm \ref{alg:2c-ane} below solves the \ane problem for monochromatic queries.

\begin{algorithm}
\caption{\label{alg:2c-ane}
Algorithm for \ane on arbitrary graphs with two colours and monochromatic queries.} 
\KwIn{A coordination game $\mathcal{G}((V,E),w,C,\beta)$ and monochromatic query $q : Q \to M$.}
\KwOut{YES if all \NEs are consistent with $q$ and NO otherwise.} 

Lines 1-9 of Algorithm \ref{alg:2c-ene} where every 0 is replaced by 1 and every 1 by 0.

{\bf if} $\forall_{i\in Q}\ s(i) = 0$ {\bf return} YES {\bf else} {\bf return} NO

\end{algorithm}

\begin{restatable}{theorem}{thmtwocolourmonoane}
\label{thm:2colour-mono-ane}
The \ane problem for coordination games with two colours and
monochromatic queries can be solved in $\calO(|G|)$ time
using Algorithm \ref{alg:2c-ane}.
\end{restatable}

In fact, any polychromatic \ane query can be reduced to two monochromatic ones and so we get the following.%
\begin{restatable}{corollary}{twocolourpolyane}
\label{cor:2colour-poly-ane}
The \ane problem for coordination games with two colours and
polychromatic queries can be solved in $\calO(|G|)$ time.
\end{restatable}

However, we will show that even answering \single \ane queries for unweighted DAGs 
is \coNP-hard in the presence of three colours and no bonuses.
We first analyse the following gadget.

\begin{figure}
\begin{minipage}{.45\textwidth}
\centering
\tikzstyle{agent}=[circle,draw=black!80,thick, minimum size=1em,scale=0.75]
\begin{tikzpicture}[auto,>=latex',shorten >=1pt,on grid]
\begin{scope}
\node[agent,label=above:{\small $\{\top, \bot\}$}](x1){\small{$X_1$}};
\node[agent, right of=x1, node distance=1.5cm,label=above:{\small $\{\top, \bot\}$}](x2){\small $X_2$};
\node[draw=none,fill=none, right of=x2,node distance=1.5cm](xdots){$\cdots$};
\node[agent, right of=xdots, node distance=1.5cm,label=above:{\small $\{\top, \bot\}$}](xk){\small $X_k$};
\node[agent, right of=xk, node distance=1.5cm, minimum size=2.3em,label=above:{$\{x\}$}](c){};
\node[agent, below of=xdots, node distance=1.5cm,label=below:{\small $\{\top, \bot\}$}](y){\large $Y$};
\draw[->] (x1) to (y);
\draw[->] (x2) to (y);
\draw[->] (xk) to (y);
\draw[->] (c) to node {{\small $k-1$}} (y);
\end{scope}
\end{tikzpicture}
\caption{\label{fig:gadget}
Gadget $D(X_1, \ldots, X_k, x; Y)$ where $x \in \{\top,\bot\}$. Note that one edge has weight $k-1$.}
\end{minipage}
\hfill
\begin{minipage}{.5\textwidth}
\centering
\tikzstyle{agent}=[circle,draw=black!80,thick, minimum size=2em,scale=0.75]
\begin{tikzpicture}[auto,>=latex',shorten >=1pt,on grid]
\newdimen\R
\R=1.1cm
\begin{scope}
\node[agent,label=above:{\small $\{\star\}$}](x1){};
\node[agent, right of=x1, node distance=\R,label=above:{\small $\{\top, \bot\}$}](x2){T};
\node[agent, right of=x2, node distance=1.2\R,label=above:{\small $\{\bot\}$}](x3){};
\node[agent, right of=x3, node distance=2.1\R,label=above:{\small $\{\top, \bot\}$}](x4){F};
\node[agent, right of=x4, node distance=\R,label=above:{\small $\{\star\}$}](x5){};
\node[agent, right of=x5, node distance=\R,label=above:{\small $\{\top,\bot\}$}](x6){$\Phi$};
\node[agent, below of=x2, node distance=\R,label=left:{\small $\{\top,\star\}$}](xy1){U};
\node[agent, below of=x3, node distance=2\R,label=left:{\small $\{\bot,\star\}$}](y1){W};
\node[agent, below of=x4, node distance=\R,label=left:{\small $\{\bot,\star\}$}](y2){X};
\node[agent, below of=x5, node distance=\R,label=right:{\small $\{\bot,\star\}$}](y3){Y};
\node[agent, right of=y1, node distance=4\R,label=right:{\small $\{\star\}$}](y4){};
\node[agent, below of=y2, node distance=2\R,label=below:{\small $\{\bot,\star\}$}](z1){Z};
\draw[->] (x1) to (xy1);
\draw[->] (x2) to (xy1);
\draw[->] (xy1) to node [above=2pt] {{\small $2$}} (y1);
\draw[->] (x3) to (y1);
\draw[->] (x4) to (y2);
\draw[->] (x5) to (y2);
\draw[->] (x5) to (y3);
\draw[->] (x6) to (y3);
\draw[->] (y1) to (z1);
\draw[->] (y2) to (z1);
\draw[->] (y3) to (z1);
\draw[->] (y4) to node [above=2pt] {{\small $2$}} (z1);
\end{scope}
\end{tikzpicture}
\caption{\label{fig:gadget3}
Gadget used in the \coNP-hardness proof of \ane. 
Edges with weight 2 can be simulated by unweighted ones.
}
\end{minipage}
\end{figure}

\begin{restatable}{proposition}{propgadget}
\label{prop:gadget}
For any \NE $s$ in $D(X_1, \ldots, X_k, x; Y)$ %
from Figure~\ref{fig:gadget}:
{\bf (a)} $s(Y) = x$ iff $\ \exists_i\ s(X_i) = x$ and
{\bf (b)} $s(Y) = \neg x$ iff $\ \forall_i\ s(X_i) = \neg x$.
\end{restatable}

Using this gadget we are able to show the following.

\begin{restatable}{theorem}{DAGtwo}
\label{thm:DAG2}
The \ane problem for \single queries is \coNP-complete for unweighted DAGs with three colours and no bonuses.
\end{restatable}
\begin{proof}
We reduce from the tautology problem for formulae in 3-DNF form. 
Assume we are given a formula 
\[\phi = (a_1 \wedge b_1 \wedge c_1) \vee (a_2 \wedge b_2 \wedge c_2) \vee \ldots \vee (a_k \wedge b_k \wedge c_k)\]
with $k$ clauses and $n$ propositional variables $x_1, \ldots, x_n$,
where each $a_i,b_i,c_i$ is a literal equal to $x_j$ or $\lnot x_j$
for some $j$. We will construct a coordination game $\game$ of size
$\mathcal{O}(n+k)$ such that a particular \single \ane query is true for $\game$ 
iff $\phi$ is a tautology.

First %
for every propositional variable $x_i$ there are four
nodes $X_i$, $\neg X_i$, $L_i$, $\overline{L}_i$ 
in $\game$, each with two possible colours $\top$ or $\bot$.
We connect these four nodes using gadgets
$D(X_i, \neg X_i, \top; L_i)$ and 
$D(X_i, \neg X_i, \bot; \overline{L}_i)$.
This makes sure that in any \NE, s, 
we have $s(L_i) = \top$ and $s(\overline{L}_i) = \bot$ 
iff $X_i$ and $\neg X_i$ are assigned different colours.
Next, for every clause $(a_i \wedge b_i \wedge c_i)$ in $\phi$ we
add to the game graph $\game$ node $C_i$. 
We use gadget $D(a_i, b_i, c_i, \bot; C_i)$ to connect 
literals with clauses, 
where we identify each $x_i$ with $X_i$ and 
each $\neg x_i$ with $\neg X_i$.
Note that Proposition \ref{prop:gadget} implies that
the colour of $C_i$ is $\top$ iff all nodes $a_i, b_i, c_i$ are assigned $\top$.
We add two nodes $T$ and $F$ to gather 
colours $\top$ and $\bot$ from the $L_i$ and $\overline{L}_i$ nodes.
Also, we add an additional node $\Phi$ to gather the values of all the clauses.
We connect these using gadgets 
$D(L_1,\ldots,L_n,\bot; T)$,
$D(\overline{L}_1,\ldots,\overline{L}_n,\top; F)$,
and
$D(C_1,\ldots,C_k,\top; \Phi)$.
Now, we need to express 
that for every \NE $s$: $s(T) = \top$ and
$s(F) = \bot$ implies that $s(\Phi) = \top$. For this we use the gadget
from Figure~\ref{fig:gadget3}. It includes three nodes $T,
F, \Phi$ that we already defined in $\game$.
 We claim that \ane query $q(Z) = \star$ is true for $\game$ iff
$\Phi$ is a tautology. {\em (The full proof is in the appendix.)}
\qed
\end{proof}

On the other hand, we can show
that answering polychromatic \ene queries is \NP-hard for unweighted DAGs 
even with two colours and no bonuses. The construction %
is similar to the one in the proof of Theorem~\ref{thm:DAG2}.

\begin{restatable}{theorem}{thmDAG}
\label{thm:DAG}
The \ene problem is \NP-complete for unweighted DAGs with two colours and no bonuses.
\end{restatable}

Building on this we can show the following when there are three colours to choose from.

\begin{restatable}{corollary}{corDAG}
\label{cor:DAG}
The \ene problem for \single queries
is \NP-complete for unweighted DAGs with three colours and no bonuses.
\end{restatable}

Note that we can also show \NP/\coNP-hardness for DAGs with 
out-degree at most two, 
because we can make arbitrary number of copies of any given node, 
e.g. to make three copies $i_1, i_2, i_3$ of node $i$ we can add nodes $i', i_1, i_2, i_3$ and edges $i \to i_1$, $i \to i'$, $i' \to i_2$, $i' \to i_3$.

\section{Simple Cycles}
\label{sec:simple-cycles}
We consider here coordination games whose underlying graph
is a simple cycle. To fix the notation, suppose that
$V=\{0,1,\ldots,n-1\}$ and the underlying graph is $0 \to 1 \to \cdots \to
n-1 \to 0$. We assume that the counting is done in cyclic order within
$\{0, \ldots, n-1\}$ using the increment operation $i \oplus 1$ and the
decrement operation $i \ominus 1$. In particular, $(n-1) \oplus 1 = 0$
and $0 \ominus 1 = n-1$.

\newcommand{\wgtone}[1]{\wgt{#1\ominus 1}{#1}}
For $i \in V$, let $Z_i(w) = \{c \in C(i) \mid \beta(i,c) + w \geq
\beta(i,c')$ for all $c' \in C(i)\}$ denote the set of colours available to player $i$
with the bonus at most $w$ below the maximum one available to $i$. 
For every $i \in V$, define $A_i := Z_i(0)$, i.e. all colours with the maximum bonus, 
$B_i := Z_i(\wgtone{i}-1)$, and
$C_i := Z_i(\wgtone{i})$.
Obviously $\emptyset \neq A_i \subseteq B_i \subseteq C_i \subseteq C(i)$ for every $i$.
It is quite easy to see that in any NE 
player $i$ can only select a colour from $C_i$.
Let us fix a query $q: Q \to M$.
In this section, without loss of generality,
we assume that $0 \in Q$ (if $0 \not\in Q$, then we can always
re-label the nodes in the cycle).

\begin{algorithm}
\caption{\label{alg:cycle-ene}
\ene on a simple cycle} 
\KwIn{A simple cycle on nodes $\{0,\ldots,n-1\}$, sets $A_i$, $B_i$, $C_i$ for $i \in V$, a query $q : Q \to M$.}
\KwOut{YES if there exists a Nash equilibrium consistent with $q$ and
  NO otherwise.}

Let $X_0 = \{q(0)\}$.

\For{$i = 0$ {\bf to} $n-1$}{
  \If{$X_i \not\subseteq B_{i\oplus 1}$}
     {$X_{i\oplus 1} = (X_i \cap C_{i\oplus 1}) \cup A_{i\oplus 1}$}
  \Else{
     {$X_{i\oplus 1} = X_i$}     
  }
  \If{$i\oplus 1 \in Q$}{
	  \If{$q(i\oplus 1) \not\in X_{i \oplus 1}$}
		  {{\bf return} NO}
	  \Else{
		$X_{i \oplus 1} = \{q(i\oplus 1)\}$
	  }
  }

}

{\bf return} YES

\end{algorithm}
\vskip-0.3em
\begin{algorithm}%
\caption{
\label{alg:cycle-ane-weighted}
Algorithm for \ane on a simple cycle.}  
\KwIn{A simple cycle on nodes $\{0,\ldots,n-1\}$, sets $A_i$, $B_i$, $C_i$ for $i \in V$, a query $q : Q \to M$.}
\KwOut{YES if all NEs are consistent with $q$ and NO otherwise.}

\For{$c \in M$}{
   \If {Algo.~\ref{alg:cycle-ene} for $q':=\{0 \to c\}$ returns NO}{
        {{\bf continue} with the next $c$}
   }\Else{
     Consider $X_{i}$ computed by Algo. \ref{alg:cycle-ene} for $q'$: \\
     \If{exists $i \in Q$ such that $X_{i} \neq \{q(i)\}$}{
        {{\bf return} NO}
     }    
   }
}
{\bf return} YES

\end{algorithm}

\begin{restatable}{theorem}{enesimplecycles}
\label{thm:ene-simple-cycles}
The \ene problem for simple cycles can be solved in $\calO(|G|)$ time.
\end{restatable}
\begin{proof}[sketch]
We argue that Algorithm \ref{alg:cycle-ene} solves the \ene problem
for simple cycles. In other words, we argue that given a simple cycle
over the nodes $V=\{0,\ldots,n-1\}$ and a query $q: Q \to M$, the
output of Algorithm \ref{alg:cycle-ene} is YES iff there exists a Nash
equilibrium $s^*$ which is consistent with $q$. Suppose there exists a
Nash equilibrium $s^*$ which is consistent with $q$. We can argue by
induction on $V$ that on termination of Algorithm \ref{alg:cycle-ene},
for all $i \in V$, we have $s^*(i) \in X_i$.

Conversely, suppose the output of Algorithm \ref{alg:cycle-ene} is
YES. From the definition, this implies that for all $i \in V$,
$X_i \neq \emptyset$ and for all $j \in Q$: $q(j) \in X_j$ (in fact,
$X_j =\{q(j)\}$). We define a Nash equilibrium $s^*$ as
follows. First, let $s^*(0)=q(0)$.  Next we assign values to $s^*(i)$
starting at $i = n-1$ and going down to $i=1$ as described below.
\begin{itemize}
\itemsep0em 
\item If $i \in Q$ then $s^*(i)=q(i)$.
\item If $i \not\in Q$ and $X_{i} \subseteq B_{i\oplus 1}$
  then by Algorithm \ref{alg:cycle-ene} we have $X_{i} = X_{i \oplus 1}$. Let $s^*(i) = s^*(i \oplus 1)$.
\item Assume $i \not\in Q$ and $X_{i} \not\subseteq
  B_{i \oplus 1}$. If $s^*(i\oplus 1) \in X_{i} \cap C_{i \oplus 1}$ set
  $s^*(i)=s^*(i\oplus 1)$. Otherwise $s^*(i\oplus 1) \in A_{i\oplus 1}$ and we set
  $s^*(i)$ to any element in $X_{i} \setminus B_{i\oplus 1}$. 
\end{itemize}
 A proof that $s^*$ is a NE is in the appendix.
\qed
\end{proof}

Algo.~\ref{alg:cycle-ane-weighted} reduces the \ane problem to $m$ \ene queries.

\begin{restatable}{theorem}{anesimplecycles}
\label{thm:ane-simple-cycles}
The \ane problem for simple cycles (unweighted simple cycles) can be solved in $\calO(m|G|)$ time (respectively, $\calO(|G|)$ time using Algorithm \ref{alg:cycle-ane} in the appendix).
\end{restatable}

\section{Colour Complete Graphs}
\label{sec:complete-graphs}

We show that \ene and \ane problems can be solved in polynomial time 
for coordination games $\mathcal{G}((V,E),C)$ played on unweighted colour complete graphs with $n$ nodes and a fixed number of colours, $m$, and no bonuses.

\begin{theorem} \label{thm:complete}
The \ene and \ane problems for unweighted colour complete graphs and
no bonuses can be solved in $\calO(nm\cdot m!)$ time.
\end{theorem}
\begin{proof}
We claim that the set of total orders on the set of colours induces
a set of joint strategies which contains the whole set $\neG$.
Specifically, every total order $\succeq$ on $M$ will be mapped to a joint strategy
$\SP(\succeq)$ as follows: assign to each player the highest colour
available to him according to the total order $\succeq$.  Formally,
for all players $i$: $\SP(\succeq)(i) = \max_{\succeq} C(i)$.
For any \NE $s$ let us define a relation $\succ_s \subseteq M \times M$: 
$x \succ_s y$ iff there exists player $i$ such that $\{x,y\} \subseteq C(i)$ and $s(i) = x$. 

\begin{restatable}{lemma}{acyc}
\label{lem:acyc}
The relation $\succ_s$ is acyclic, i.e. for all $k \geq 2$
there is no sequence of colours $x_1,\ldots,x_k$ such that
$x_1 \succ_s x_2 \succ_s \ldots \succ_s x_k \succ_s x_1$.
\end{restatable}

Note
Lemma \ref{lem:acyc} may fail when bonuses are introduced into the game. 
We also need the following folk result.

\begin{restatable}{lemma}{acycorder}
\label{lem:acyc-order}
Any acyclic binary relation on a finite set can be extended to a total oder.
\end{restatable}
For the relation $\succ_s$ let $\succeq^*_s$ be a total order from Lemma \ref{lem:acyc-order} such that
$\succ_s\ \subseteq\ \succeq^*_s$.

\begin{restatable}{lemma}{lemordertwone}
\label{lem:order-2-ne}
For any \NE s,  $\SP(\succeq^*_s)$ $= s$. 
\end{restatable}
From Lemma \ref{lem:acyc} and Lemma \ref{lem:order-2-ne} we know that 
for every \NE $s$, there exists at least one total order on $M$
that induces it.
Therefore, for \ene problem (\ane problem) it suffices to check for
all possible total orders $\succeq$ on $M$, whether the induced joint
strategy $\SP(\succeq)$, is a \NE and if so, whether any
(respectively, all) of them is consistent with $q$.  There are $m!$
total orders on $M$. Checking whether an induced strategy profile 
is a \NE consistent with $q$ takes $\calO(nm)$ time. This gives
$\calO(nm\cdot m!)$ in total.
\qed
\end{proof}

Note that there are coordination games on colour complete graphs 
with one-to-one correspondence between 
the set of total orders on colours and the set of all \NEs
(Example \ref{ex:many-NEs} in the appendix), 
and so with exponentially many different NEs.

\section{Directed Acyclic Graphs}
\label{sec:pseudoforests}

In Section~\ref{sec:two-colours} we showed that the \ene and \ane
problems are \NP and \coNP complete respectively even for unweighted DAGs
with out-degree at most two and no bonuses. We now
show that if the out-degree of each node in an unweighted DAG is at most 1
(there are no constraints on the in-degree of nodes) then these
problems can be solved efficiently. 

\begin{algorithm}
\caption{
\label{alg:out1e}
Algorithm for \ene on unweighted DAGs with out-degree $\leq 1$.} 
\KwIn{A coordination game $\mathcal{G}((V,E),C,\beta)$ and query $q : Q \to M$} 
\KwOut{YES if there exists a \NE consistent with $q$ and NO otherwise.} 

Topologically sort $V$ into a sequence $(i_1,\ldots,i_n)$.

\For{$j := 1\ldots n$}{
	$X(i_j):=\emptyset$\\
	$Y := \{X(k)\ |\ k \to i_j \in E \}$\\
	\For{$c \in C(i_j)$}{
		$S := \{Z \in Y\ |\ c \in Z\};\ \ \ $
		$C' := C \setminus\{c\};\ \ \ $ 
		$Y' := Y \setminus S;\ \ \ $\\ 
		\If{exists $c' \in C'$ such that $|S| + \beta(i_j, c) - \beta(i_j, c') < 0$}
		   {{\bf continue} with the next $c$}
		\While{exists $c' \in C'$ such that $|S| + \beta(i_j, c) - \beta(i_j, c') \geq |Y'|$}{
		    $C' := C' \setminus \{c'\};\ \ \ \ $ 
		    $Y' := Y' \setminus \{Z \in Y'\ |\ c' \in Z \}$
		}
		Construct the following bipartite graph
{\openup-1\jot
\setlength{\abovedisplayskip}{0pt}
\setlength{\belowdisplayskip}{0pt}
\setlength{\abovedisplayshortskip}{0pt}
\setlength{\belowdisplayshortskip}{0pt}
\begin{align*}
		G' := (V' = (&Y', \{\{c'\} \times \{1,\ldots,|S| + \\
		&\beta(i_j, c) - \beta(i_j, c') \}\ |\ c' \in C'\}), E') \\
	 	&\text{where } Z\to(c',x) \in E' \text{ iff } c' \in Z
\end{align*}
}
		\If{the maximum bipartite matching in $G'$ has size $|Y'|$}
			{add $c$ to $X(i_j)$}
		}

	\If{$i_j \in Q$}{
	{\bf if} $q(i_j) \not\in X(i_j)$ {\bf return} NO {\bf else} $X(i_j):=\{q(i_j)\}$
	} 
	
}

{\bf return} YES

\end{algorithm}

\begin{restatable}{theorem}{enedagone}
\label{thm:ene-dag-out1}
Algorithm \ref{alg:out1e} solves the \ene problem for unweighted DAGs with out-degree at most one 
in $\calO(|G|^{2.5})$ time.
\end{restatable}
\begin{proof}[sketch] Intuitively, for each node, $i$, we compute the set, $X(i)$, of colours 
that can possibly be assigned to $i$ in any \NE. 
Such a set is trivial to compute for source nodes in $G$, and for the other nodes
it can be computed by constructing a suitable bipartite graph based 
on the sets precomputed for all its neighbours and running a matching algorithm.
In lines 7-10 we remove colours that are dominated by others.
We need the following lemma.
\begin{restatable}{lemma}{enedagonelemma}
\label{lm:DAG-out1-ene}
If Algorithm \ref{alg:out1e} returns YES, then for all $i \in V$, for
all $c \in X(i)$, there exists a Nash equilibrium $s^*$ such that
$s^*_i=c$ and for all $j \neq i$, $s^*_j \in X(j)$.
\end{restatable}
Now, if  Algorithm \ref{alg:out1e} returns YES, then from the
definition, for all $i \in V$, $A_i \neq \emptyset$ and for all $j \in
P$, $A_j=\{q(j)\}$. By Lemma~\ref{lm:DAG-out1-ene} it follows that
there exists a Nash equilibrium $s^*$ which is consistent with $q$.

Conversely, suppose there exists a Nash equilibrium $s^*$ which is
consistent with $q$. Let $\theta=(i_1,\ldots,i_n)$ be the topological
ordering of $V$ chosen in line 1 of Algorithm \ref{alg:out1e}. We
argue that for all $j \in \{1,\ldots,n\}$, $s^*(i_j) \in X(i_j)$. The
claim follows easily for $i_1$. Consider a node $i_m$ and suppose for
all $j < m$, $s^*(i_j) \in X(i_j)$. For $c \in C$, let $N_{i_m}(s^*,c)
=\{i_k \in N_{i_m} \mid s^*(i_k) =c\}$. Since $s^*$ is a Nash
equilibrium, $s^*(i_m)$ is a best response to the choices made by all
nodes $i_k \in N_{i_m}$. This implies that for all $c \neq s^*_{i_m}$,
$|N_{i_m}(s^*,c)| + \beta(i_{j}, c) \leq |N_{i_m}(s^*,s^*_{i_m})| + \beta(i_{j}, s^*_{i_m})$.
Note that $|S| \geq |N_{i_m}(s^*,s^*_{i_m})|$ and so $c$ is not discarded in line 8.
Also, it guarantees the existence of a matching of size $|Y'|$ at line 12 
and thus $s^*(i_m) \in X(i_m)$.

We claim that if the Hopcroft-Karp algorithm is used for each matching at line 11, 
then Algorithm \ref{alg:out1e} runs in $\calO(|G|^{2.5})$. 
First, 
for each node $k$,
$X(k)$ is in $Y$ at most once and so is matched
at most once for each colour.
We claim that the worst case running time is for $|Y| = |V|$. 
Now, due to lines 9-10 we have $|S| + \beta(i_j, c) - \beta(i_j, c') \leq |Y'|$ $=$ $\calO(n)$, 
so $G'$ at line 11 
has $\calO(nm)$ nodes, $\calO(n\cdot nm)$ edges and one matching takes $\calO(\sqrt{nm} \cdot n^2m)$ time.
\qed
\end{proof}

Similarly Algorithm \ref{alg:out1a} solves the \ane problem.

\begin{algorithm}
\caption{\label{alg:out1a}
Algorithm for \ane on unweighted DAGs with out-degree $\leq 1$.} 
\KwIn{A coordination game $\mathcal{G}((V,E),C,\beta)$ and query $q : Q \to M$.} 
\KwOut{YES if all \NEs are consistent with $q$ and NO otherwise.} 

Topologically sort $V$ into a sequence $(i_1,\ldots,i_n)$.

\For{$j := 1\ldots n$}{
	$X(i_j) :=$ the set of colours player $i_j$ can play in any \NE %
	(lines 3-13 of Algorithm \ref{alg:out1e})\\
	\If{$i_j \in Q$ {\bf and} $X(i_j) \neq \{q(i_j)\}$}{
	{\bf return} NO 
	} 

}

{\bf return} YES
\end{algorithm}

\vspace{-0.1em}
\begin{restatable}{theorem}{anedagone}
\label{thm:ane-dag-out1}
Algorithm \ref{alg:out1a} solves the \ane problem for DAGs with out-degree at most one in
$\calO(|G|^{2.5})$ time.
\end{restatable}

\section{Conclusions}
\label{sec:conclusions}
We presented a simple class of coordination games on directed graphs.
We focused on checking whether a given partial colouring of a subset
of the nodes is consistent with some pure \NE or, alternatively, with
all pure \NEs.  We showed these problems to be \NP-complete and
\coNP-complete, respectively, in general.  However, we also
identified several natural cases when these decision problems 
are tractable.

In the case of weighted DAGs with out-degree at most one and 
colour complete graphs with no bonuses
a simple reduction from the {\sc Partition} problem and its complement, shows \NP and \coNP-hardness of their \ene and \ane problems, respectively.
This does not exclude the
possibility that pseudo-polynomial algorithms exist for these problems.
We conjecture that even for unweighted 
colour complete graphs
these problems are NP/coNP-hard 
in the presence of bonuses or when the set of colours, $M$, is not fixed.

There are several ways our results can be extended further. 
One is to study other constraints, e.g. uniqueness of \NE or checking maximum payoff for a given player.
Another is to look at different solution concepts, e.g. strong equilibria.
And yet another is to look for more classes of graphs that can be
analysed in polynomial time.  Given that these decision problems are
already computationally hard for DAGs with three colours, the
possibilities for such new classes are rather limited.

Finally, we only focused on pure \NEs in this paper, which may not exist for general graphs.
On the other hand, mixed \NEs always exist due to Nash's theorem.
It would be interesting to know 
whether the complexity of finding one is
{\sc PPAD}-complete problem just like it is
for general polymatrix games \cite{CD11}.
\bibliographystyle{aaai}
\bibliography{e,clustering,extrabib}

\appendix
\clearpage

{\LARGE \bf Appendix}

\section{Full algorithms}

\setcounter{algocf}{1}
\begin{algorithm}%
\caption{
Algorithm for \ane on arbitrary graphs with two colours and monochromatic queries.} 
\KwIn{A coordination game $\mathcal{G}((V,E),w,C,\beta)$ and monochromatic query $q : Q \to M$.}
\KwOut{YES if all \NEs are consistent with $q$ and NO otherwise.} 

\For{$i \in V$}{
  {\bf if} {$1 \in C(i)$} {\bf then} {$s(i) = 1$} {\bf else} {$s(i) = 0$} 
}
{\bf set} $\mathcal{S} := \{i\ |\ s(i) = 0\}$ 

\While{$\mathcal{S} \neq \emptyset$}{
   $\text{remove any element from } \mathcal{S}$ {and assign it to $i$}\\
   \For{$\{j \in V$ | $i \to j \in E\}$}{
      \If{$s(j) = 1$ {\bf and} $0 \in C(j)$ {\bf and} $p_j((0,s_{-j}))>p_j(s)$}
      {    $s(j) = 0$ \\
           $\text{add $j$ to }\mathcal{S}$
      }     
   }
}   

{\bf if} $\forall_{i\in Q}\ s(i) = 0$ {\bf return} YES {\bf else} {\bf return} NO

\end{algorithm}

\setcounter{algocf}{3}

\setcounter{algocf}{5}
\begin{algorithm}%
\caption{
Algorithm for \ane on unweighted DAGs with out-degree $\leq 1$.} 
\KwIn{A coordination game $\mathcal{G}((V,E),C,\beta)$ and query $q : Q \to M$} 
\KwOut{YES if all \NEs are consistent with $q$ and NO otherwise.} 

Topologically sort $V$ into a sequence $(i_1,\ldots,i_n)$.

\For{$j := 1\ldots n$}{
	$X(i_j):=\emptyset$\\
	$Y := \{X(k)\ |\ k \to i_j \in E \}$\\
	\For{$c \in C(i_j)$}{
		$S := \{Z \in Y\ |\ c \in Z\};\ \ \ $
		$C' := C \setminus\{c\};\ \ \ $ 
		$Y' := Y \setminus S;\ \ \ $\\ 
		\If{exists $c' \in C'$ such that $|S| + \beta(i_j, c) - \beta(i_j, c') < 0$}
		   {{\bf continue} with the next $c$}
		\While{exists $c' \in C'$ such that $|S| + \beta(i_j, c) - \beta(i_j, c') \geq |Y'|$}{
		    $C' := C' \setminus \{c'\};\ \ \ \ $ 
		    $Y' := Y' \setminus \{Z \in Y'\ |\ c' \in Z \}$
		}
		Construct the following bipartite graph
\begin{align*}
		G' := (V' = (&Y', \{\{c'\} \times \{1,\ldots,|S| + \\
		&\beta(i_j, c) - \beta(i_j, c') \}\ |\ c' \in C'\}), E') \\
	 	&\text{where } Z\to(c',x) \in E' \text{ iff } c' \in Z
\end{align*}
		\If{the maximum bipartite matching in $G'$ has size $|Y'|$}
			{add $c$ to $X(i_j)$}
		}
	\If{$i_j \in Q$ {\bf and} $X(i_j) \neq \{q(i_j)\}$}{
	{\bf return} NO 
	} 
}

{\bf return} YES
\end{algorithm} 

\begin{algorithm}
\caption{
\label{alg:cycle-ane}
Algorithm for \ane on an unweighted simple cycle.}  
\KwIn{A simple cycle on nodes $V=\{0,\ldots,n-1\}$, sets $A_i$, $B_i$, $C_i$ for $i \in V$, and a query $q : Q \to M$.}
\KwOut{YES if all NEs are consistent with $q$ and NO otherwise or if no NE exists.}

Let $X_0 = \{q(0)\}$.

\For{$i = 0$ {\bf to} $n-1$}{
  \If{$X_i \not\subseteq B_{i\oplus 1}$}
     {$X_{i\oplus1} = (X_i \cap C_{i\oplus 1}) \cup A_{i\oplus 1}$}
  \Else{
     {$X_{i\oplus1} = X_i$}     
  }
  \If{$i\oplus 1 \in Q$}{
	  \If{$\{q(i\oplus 1)\} \neq X_{i\oplus1}$}
		  {{\bf return} NO}
  }

}

{\bf return} YES

\end{algorithm}

\section{Full proofs of lemmas and theorems}
\begin{figure}%
\centering
\tikzstyle{agent}=[circle,draw=black!80,thick, minimum size=1.3em]
\begin{tikzpicture}[auto,>=latex',shorten >=1pt,on grid,scale=1]
\begin{scope}
\node[agent,label=below:{\small $\{a\}$}](x1){};
\node[agent, left of=x1, node distance=1.5cm,label=below:{\small $\{a,b\}$}](x2){};
\node[agent, left of=x2, node distance=1.5cm,label=below:{\small $\{b\}$}](x3){};
\draw[->] (x1) to [bend left] (x2);
\draw[->] (x2) to [bend left] (x1);
\draw[->] (x2) to [bend left] (x3);
\draw[->] (x3) to [bend left] (x2);
\end{scope}
\end{tikzpicture}
\caption{\label{fig:colour-complete-graph}
A graph which is colour complete, but is not a complete graph (a clique).
}
\end{figure}
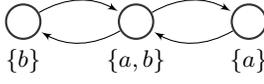

\twocolourmonoene*
\begin{proof}
We show that Algorithm \ref{alg:2c-ene} solves the \ene problem and that
its running time is $\calO(|G|)$.
Let $\preceq$ be a partial order on all joint strategies
$s : V \to M$ defined as follows: $s \preceq s'$ iff for all $i \in
V$, $s(i) \leq s'(i)$.  Let $s_0$ denote the value of $s$ once line 3
is reached.  The colouring $s_0$ may not be a \NE, so
Algorithm \ref{alg:2c-ene} tries to correct this with the minimum
number of switches from $0$ to $1$.  Note that for any colouring
$s$ we have $s_0 \preceq s$.  Note that lines 3-9 of
Algorithm \ref{alg:2c-ene} can be seen as a function $F : (V \to
M) \to (V \to M)$ from the initial colouring, in this case $s_0$, to a
new colouring, $F(s_0)$.  Note that $F$ is monotonic according to
$\preceq$, i.e. if $s \preceq s'$ then $F(s) \preceq F(s')$. This is
simply because the more colour~$1$ is used initially, the more players
would like to switch to it.  Also, any \NE is a fixed point of $F$,
because no player would like to switch at line 7.
We now need the following lemma. %
\begin{restatable}{lemma}{lmcolourext}
\label{lm:colour-ext}
For every joint strategy $s$, $F(s)$ is a \NE.
\end{restatable}
\begin{proof}
Every node with colour $1$ in $F(s)$ is added to the set $\mathcal{S}$ at most once:
either at the beginning %
or when it switches from $0$ to $1$. 
If a node does not have a predecessor with colour $1$, 
it cannot possibly have an incentive to switch to $1$, because
this would give him reward $0$.
Every time a predecessor of a node switches to $1$, 
we consider that node in line 7 and whether it is beneficial 
for it to switch to $1$.
If at no point it was, then colour $0$ has to be this player's best response in $F(s)$.
Also, no player can have an incentive to switch back from $1$ to $0$ because
the payoff for choosing $1$ is weakly increasing for every player after each strategy update.
\qed
\end{proof}

Now, if Algorithm \ref{alg:2c-ene} returns YES, then the correctness follows
from Lemma~\ref{lm:colour-ext}. Since in this case, $F(s_0)$ is
consistent with $q$ and by Lemma~\ref{lm:colour-ext} it is a \NE.  
Conversely, if Algorithm \ref{alg:2c-ene} returns NO then there exists 
$i\in Q$ such that $F(s_0)(i) = 1$.  Suppose there is a \NE $s'$
consistent with $q$.  Then $s_0 \preceq s'$ and $F(s_0) \preceq F(s')
= s'$, but $s'(i) = q(i) = 0$; a contradiction.

To analyse its computational complexity,
note that each node can be added to the set $\mathcal{S}$ at most once, 
because the colour of each node changes at most once
and so each edge is considered at most once as well.
Moreover, we can compute $p_j((1,s_{-j}))$ and $p_j(s)$ in constant time, 
by storing for each node the sum of weights of edges from neighbours with colour 1.
Every time the colour of a node $j$ changes in line 8, 
for any neighbour $i$ of $j$ we add the weight of the edge leading from $j$ to $i$ to the stored value for node $i$; we need to make such an update $\calO(e)$ times in total. Thus the total complexity of this algorithm is $\calO(n+e)$.\qed\end{proof}

\thmtwocolourmonoane*
\begin{proof}
Let $s_0$ be the joint strategy defined by lines 1--2 in Algorithm
\ref{alg:2c-ane}. By an argument very similar to the proof of Lemma 1, we
can show that $F(s_0)$ is a Nash equilibrium. If Algorithm
\ref{alg:2c-ane} returns NO then there exists a $j \in Q$ such that
$F(s_0)(j) \neq q(j)$. Therefore, $F(s_0)$ is a Nash equilibrium which
is not consistent with $q$.

To show the converse, as in the proof of
Theorem~\ref{thm:2colour-mono-ene}, we define a partial
order $\preceq$ on joint strategies as before. 
Note that for any joint strategy $s$ we have $s \preceq s_0$. 
Again, note that lines 3--9 of Algorithm \ref{alg:2c-ane} define a
function $F:(V \to M) \to (V \to M)$ which satisfies the property: if
$s \preceq s'$ then $F(s) \preceq F(s')$.

Now suppose that Algorithm \ref{alg:2c-ane} returns YES then for all $i
\in Q$: $F(s_0)(i)=0$. We need to prove that every Nash equilibrium
is consistent with $q$. Suppose this is not the case, then there
exists a Nash equilibrium $s'$ and a node $j \in Q$ such that $s'(j)
\neq q(j)$. By our assumption, this implies that $s'(j)=1$. We have
$s' \preceq s_0$ and therefore $s' = F(s') \preceq F(s_0)$. From
$s'(j)=1$ and $F(s_0)(j)=0$ we get a contradiction.

The time complexity analysis of Algorithm~\ref{alg:2c-ane} is the same as
that of Algorithm~\ref{alg:2c-ene}.  
\qed
\end{proof}

\twocolourpolyane*
\begin{proof}
Let $q: Q \to M$ be a polychromatic query.
Define $P_0$ and $P_1$ to be the sets of players asked to
pick $0$ and $1$, respectively, by $q$.
Formally, $P_0 = \{i \in Q \mid q(i) = 0\}$ and $P_1 = \{i \in Q \mid q(i) = 1\}$.
Let $q_0 = q|_{P_0}$ and $q_1 = q|_{P_1}$.
It can be verified that
every Nash equilibria is consistent with $q$ iff
every Nash equilibria is consistent with $q_0$ and $q_1$. Note that both $q_0$ and $q_1$
are monochromatic queries and therefore, by
Theorem~\ref{thm:2colour-mono-ane}, both of them can be answered in 
$\calO(|G|)$ time. Thus the claim follows.\qed
\end{proof}

\propgadget*
\begin{proof}
{\bf (a)} If $\exists_i\ s(X_i) = x$ then player $Y$'s payoff 
for picking $x$ is at least $k$ and for $\neg x$ is at most $k$, 
so it has to be $s(Y) = x$.
On the other hand, if $\forall_i\ s(X_i) \neq x$ then 
player $Y$'s payoff for picking $x$ is $k-1$
and for picking $\neg x$ is $k$, so it has to be $s(Y) \neq x$.

\noindent {\bf (b)} If $\forall_i\ s(X_i) = \neg x$ then player $Y$'s payoff 
for picking $x$ is $k-1$ and for $\neg x$ is $k$, 
so it has to be $s(Y) = \neg x$.
On the other hand, if $\exists_i\ s(X_i) \neq x$ then 
player $Y$'s payoff for picking $x$ is at least $k$ 
and for picking $\neg x$ is at most $k-1$, so it has to be $s(Y) \neq x$.\qed
\end{proof}

\DAGtwo*
\begin{proof}
  To prove \coNP-hardness we provide a reduction from the tautology problem for formulae in 3-DNF form,
  which is \coNP-complete.  Assume we are given a 3-DNF formula
\[\phi = (a_1 \wedge b_1 \wedge c_1) \vee (a_2 \wedge b_2 \wedge c_2) \vee \ldots \vee (a_k \wedge b_k \wedge c_k)\]
with $k$ clauses and $n$ propositional variables $x_1, \ldots, x_n$,
where each $a_i,b_i,c_i$ is a literal equal to $x_j$ or $\lnot x_j$
for some $j$. We will construct a coordination game $\game$ of size
$\mathcal{O}(n+k)$ such that a particular \single \ane query is true for $\game$ 
iff $\phi$ evaluates to true for all truth assignments.

First %
for every propositional variable $x_i$ there are four
nodes $X_i$, $\neg X_i$, $L_i$, $\overline{L}_i$ 
in $\game$, each with two possible colours $\top$ or $\bot$.
We connect these four nodes using gadgets
$D(X_i, \neg X_i, \top; L_i)$ and 
$D(X_i, \neg X_i, \bot; \overline{L}_i)$.
This makes sure that in any \NE, s, 
we have $s(L_i) = \top$ and $s(\overline{L}_i) = \bot$ 
iff $X_i$ and $\neg X_i$ are assigned different colours.

Next, for every clause $(a_i \vee b_i \vee c_i)$ in $\phi$ we
add to the game graph $\game$ node $C_i$. 
We use gadget $D(a_i, b_i, c_i, \bot; C_i)$ to connect 
literals with clauses, 
where we identify each $x_i$ with $X_i$ and 
each $\neg x_i$ with $\neg X_i$.
Note that Proposition \ref{prop:gadget} implies that
the colour of $C_i$ is $\top$ iff all nodes $a_i, b_i, c_i$ are assigned $\top$.

We add two nodes $T$ and $F$ to gather 
colours $\top$ and $\bot$ from the $L_i$ and $\overline{L}_i$ nodes.
Also, we add an additional node $\Phi$ to gather the values of all the clauses.
We connect these using gadgets 
$D(L_1,\ldots,L_n,\bot; T)$,
$D(\overline{L}_1,\ldots,\overline{L}_n,\top; F)$,
and
$D(C_1,\ldots,C_k,\top; \Phi)$.
The first two gadgets guarantee that if in a \NE $s$ the colour of $T$ is $\top$ 
and the colour of $F$ is $\bot$ then
$s$ corresponds to a valid truth assignment.
The last gadget guarantees that the colour of $\Phi$ is $\top$ iff
at least one of $C_i$-s has colour $\top$.

Now, we need to express that for every \NE $s$:
 $s(T) = \top$ and $s(F) = \bot$ implies that $s(\Phi) = \top$.
We will use gadget from Figure~\ref{fig:gadget3}.
It consists of the three nodes $T, F, \Phi$ that we already defined in $\game$
and several additional ones. We claim that \ane query $q(Z) = \star$ is true
for $\game$ iff $\Phi$ is a tautology.
However, equivalently, we will prove that
\ane query $q(Z) = \star$ is false
for $\game$ iff $\phi$ is not a tautology. 

\smallskip \noindent ($\Rightarrow$) 
Let $s$ be a \NE which does not satisfy query $q(Z) = \star$, 
which essentially means that $s(Z) = \bot$.
We will show that the following truth assignment 
$\nu(x_i) = s(X_i)$ makes $\phi$ false.
Looking at the gadget in Figure~\ref{fig:gadget3} we can 
easily deduce that all nodes $W,X,Y$ are assigned $\bot$ in $s$, because
otherwise $Z$ would have an incentive to switch to $\star$.
This means that it has to be 
$s(X) = s(\Phi) = \bot$, 
and $s(U) = \top$ so $s(T) = \top$.
Next, $s(T) = \top$ implies that $s(L_i) = \top$ for all $i$ and
$s(F) = \bot$ implies that $s(\overline{L}_i) = \bot$ for all $i$, 
so $s(\neg X_i) = \neg s(X_i)$ for all $i$.
Finally, $s(\Phi) = \bot$ implies that $s(C_i) = \bot$ for all $i$,
but then $\nu$ makes every clause in $\phi$ false, 
and so also makes the whole formula $\phi$ false.

\smallskip \noindent ($\Leftarrow$) 
Let $\nu :\{x_1, \ldots, x_n\} \to \{\top,\bot\}$
be a truth assignment that makes $\phi$ false.
We form the following \NE, $s$, by first setting 
$s(X_i) = \nu(x_i)$ and $s(\neg X_i) = \neg \nu(x_i)$ for all $i$.
Note that this makes the best response of nodes $L_i$ to be
$\top$ and of nodes $\overline{L}_i$ to be $\bot$.
It follows that the best responses of $T$ and $F$ are $\top$ and $\bot$, 
respectively.
On the other hand, since $\nu$ makes $\phi$ false, 
all clauses $C_1, \ldots, C_n$ in $\phi$ are false,
and so for all $i$: $s(C_i) = \bot$ is $C_i$'s best response.
Finally, the best response of node $\Phi$ is $\bot$.
Looking at the gadget in Figure~\ref{fig:gadget3},
given the values $s(T) = \top, s(F) = s(\Phi) = \bot$, 
one can easily see that 
$s(U) = \top$, $s(W) = s(X) = s(y) = s(Z) = \bot$
are these nodes best responses.
Therefore, $s$ is a \NE which does not satisfy query $q(Z) = \star$.
\qed
\end{proof}

\thmDAG*
\begin{proof}
  To prove NP-hardness we provide a reduction from the 3-SAT problem,
  which is NP-complete.  Assume we are given a 3-SAT formula
\[\phi = (a_1 \vee b_1 \vee c_1) \wedge (a_2 \vee b_2 \vee c_2) \wedge \ldots \wedge (a_k \vee b_k \vee c_k)\]
with $k$ clauses and $n$ propositional variables $x_1, \ldots, x_n$,
where each $a_i,b_i,c_i$ is a literal equal to $x_j$ or $\lnot x_j$
for some $j$. We will construct a coordination game $\game$ of size
$\mathcal{O}(n+k)$ such that a particular \ene query is true for $\game$ 
iff $\phi$ is satisfiable.

First, for every propositional variable $x_i$ there are four
nodes $X_i$, $\neg X_i$, $L_i$, $\overline{L}_i$ 
in $\game$, each with two possible colours $\top$ or $\bot$.
Intuitively, for a given truth assignment, if $x_i$ is true then
$\top$ should be chosen for $X_i$ 
and $\bot$ should be chosen for $\neg X_i$, 
and the other way around if $x_i$ is false.
To select only the \NEs which correspond to valid truth assignments 
we make use of the gadget $D$ presented in Figure~\ref{fig:gadget}.
We connect these four nodes using gadgets
$D(X_i, \neg X_i, \top; L_i)$ and 
$D(X_i, \neg X_i, \bot; \overline{L}_i)$.
This make sure that in any \NE, s, 
we have $s(L_i) = \top$ and $s(\overline{L}_i) = \bot$ 
iff $X_i$ and $\neg X_i$ are assigned different colours.
This is because from Proposition \ref{prop:gadget} it follows that 
if $s(L_i) = \top$ then 
$\top$ is assigned to at least one of $X_i,\neg X_i$ and 
if $s(\overline{L}_i) = \bot$ then 
$\bot$ is assigned to at least one of them as well.
So necessarily, $\top$ and $\bot$ are assigned to exactly one of them.

Next, for every clause $(a_i \vee b_i \vee c_i)$ in $\phi$ we
add to the game graph $\game$ node $C_i$. 
We use gadget $D(a_i, b_i, c_i, \top; C_i)$ to connect 
literals with clauses, 
where we identify each $x_i$ with $X_i$ and 
each $\neg x_i$ with $\neg X_i$.
Note that Proposition \ref{prop:gadget} implies that
$s(C_i) = \top$ iff at least one of nodes $a_i, b_i, c_i$ 
is assigned $\top$.

Finally, we have two nodes $T$ and $F$ which gather 
all nodes whose colours should be $\top$ and $\bot$, respectively.
We connect these using gadgets 
$D(L_1,\ldots,L_n,C_1,\ldots,C_k,\bot; T)$
and
$D(\overline{L}_1,\ldots,\overline{L}_n,\top; F)$.

We claim that \ene query $q(T) = \top, q(F) = \bot$ is true
for $\game$ iff $\phi$ is satisfiable.
 
\smallskip \noindent ($\Rightarrow$) Assume that $s$ is a \NE
consistent with $q$ in the game $\game$.  We claim that the truth
assignment $\nu :\{x_1, \ldots, x_n\} \to \{\top,\bot\}$ that assigns
$\nu(x_i) = \top$ iff $s(X_i) = \top$, 
and $\nu(x_i) = \bot$ iff $s(\neg X_i) = \top$, 
makes $\phi$ true.  

Since $s$ is a \NE and $s(T) = \top$, Proposition \ref{prop:gadget} implies
that all $L_i$-s and $C_i$-s are assigned colour $\top$.
Similarly, $s(F) = \bot$ implies that all $\overline{L}_i$-s
are assigned colour $\bot$.
But this means that the assignment of the colours to 
$X_i$-s and $\neg X_i$-s corresponds to a valid truth assignment.
Furthermore, for any
$i \in \{1, \LL, k\}$: $s(C_i) = \top$ implies that 
at least one of the literals $a_i$, $b_i$, $c_i$ is assigned $\top$. 
Therefore $\nu$ makes every clause $C_i$ true and
so the whole formula $\phi$ true as well.

\smallskip \noindent ($\Leftarrow$) Assume $\phi$ is satisfiable. Take
a truth assignment $\nu :\{x_1, \ldots, x_n\} \to \{\top,\bot\}$ that
makes $\phi$ true. We will construct a \NE $s$ consistent with $q$.
For all $j$, if $\nu(x_j)$ is true then assign 
$s(X_i) = \top$, $s(\neg X_i) = \bot$, 
and if $\nu(x_j)$ is false then assign
$s(X_i) = \bot$, $s(\neg X_i) = \top$.
It follows that if we assign
$s(L_i) = \top$, $s(\overline{L}_i) = \top$ for all $i = 1,\ldots,n$
then $L_i$ and $\overline{L}_i$ have no incentive to switch.
Furthermore, because $\nu$ makes every clause $C_i$ true, 
$\top$ is assigned in $s$ to at least one of the nodes $a_i$, $b_i$, $c_i$,
so if we set $s(C_i) = \top$ for all $i = 1,\ldots,k$, 
then no $C_i$ has an incentive to switch.
Finally, setting
$s(T) = \top$ and $s(F) = \bot$, 
makes $s$ consistent with $q$ and 
neither $T$ nor $F$ has an incentive to switch.
\qed
\end{proof}

\begin{figure}%
\centering
\tikzstyle{agent}=[circle,draw=black!80,thick, minimum size=2em,scale=0.9]
\begin{tikzpicture}[auto,>=latex',shorten >=1pt,on grid]
\begin{scope}
\node[agent,label=above:{\small $\{\bot\}$}](x1){};
\node[agent, right of=x1, node distance=1.5cm,label=above:{\small $\{\top, \bot\}$}](x2){T};
\node[agent, right of=x2, node distance=1.5cm,label=above:{\small $\{\star\}$}](x3){};
\node[agent, right of=x3, node distance=1.5cm,label=above:{\small $\{\top, \bot\}$}](x4){F};
\node[agent, right of=x4, node distance=1.5cm,label=above:{\small $\{\top$\}}](x5){};
\node[agent, below of=x2, node distance=1.5cm,label=below:{\small $\{\bot,\star\}$}](y1){X};
\node[agent, below of=x3, node distance=1.5cm,label=above:{\small $\{\top\}$}](y2){};
\node[agent, below of=x4, node distance=1.5cm,label=below:{\small $\{\top,\star$\}}](y3){Y};
\node[agent, below of=y2, node distance=1.5cm,label=below:{\small $\{\top,\star$\}}](z1){Z};
\draw[->] (x1) to (y1);
\draw[->] (x2) to (y1);
\draw[->] (x3) to (y1);
\draw[->] (x3) to (y3);
\draw[->] (x4) to (y3);
\draw[->] (x5) to (y3);
\draw[->] (y1) to (z1);
\draw[->] (y2) to node [left] {{\small $2$}} (z1);
\draw[->] (y3) to (z1);
\end{scope}
\end{tikzpicture}
\caption{\label{fig:gadget2}
Gadget used in the \NP-hardness proof of \ene. Note that there is one edge with weight 2, which can be easily simulated by unweighted edges.}
\end{figure}
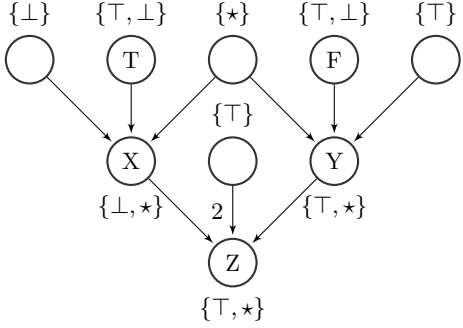
\corDAG*
\begin{proof}
To prove NP-hardness we again reduce from the 3-SAT problem.
Assume we are given a 3-SAT formula $\phi$. 
In Theorem \ref{thm:DAG} we constructed a game $\game$ for which
$\phi$ is satisfiable iff \ene query $q(T) = \top, q(F) = \bot$ is true for $\game$, 
where $T$ and $F$ are two nodes of $\game$.
We now combine this reduction with the gadget depicted in Figure~\ref{fig:gadget}, 
which consists of several nodes including nodes $T$ and $F$ from $\game$,
to form a new game $\game'$. We claim that a \single query $q(Z) = \star$ is true in $\game'$
iff $\phi$ is satisfiable. 

\smallskip \noindent ($\Rightarrow$) 
Let $s$ be a \NE satisfying $s(Z) = \star$.
Notice that based on the structure of the gadget, 
$s(Z) = \star$ implies that $s(X) = s(Y) = \star$, 
which implies that $s(T) = \top$ and $s(F) = \bot$. 
We already know that this implies that $\phi$ is satisfiable.

\smallskip \noindent ($\Leftarrow$) 
If $\phi$ is satisfiable then there exists a \NE $s$ in $\game$ such that
$s(T) = \top$ and $s(F) = \bot$. 
Notice that $s$ can easily be extended to a \NE $s'$ in $\game'$ by
setting $s'(X) = s'(Y) = s'(Z) = \star$, which is consistent with the query $q$.
\qed
\end{proof}

\enesimplecycles*
\begin{proof}
We show that given a simple cycle over the nodes $V=\{0,\ldots,n-1\}$
and a query $q: Q \to M$, the output of Algorithm \ref{alg:cycle-ene}
is YES iff there exists a Nash equilibrium $s^*$ which is consistent
with $q$.

Suppose there exists a Nash equilibrium $s^*$ which is consistent with
$q$. We can argue by induction on $n$ that on termination of
Algorithm \ref{alg:cycle-ene}, for all $i \leq n$, we have $s^*(i) \in
X_i$, which in turn implies that the output of
Algorithm \ref{alg:cycle-ene} is YES. Since $s^*$ is consistent with
$q$, we have $s^*(0) = q(0)$ and by line 1 of
Algorithm \ref{alg:cycle-ene}, $X_0 =
\{s^*(0)\}$. Assume that we have $s^*(i) \in X_i$ and consider the
iteration of the loop in line 2 of Algorithm \ref{alg:cycle-ene} for
$i \oplus 1$.  We have the following cases.

If $X_i \subseteq B_{i \oplus 1}$ then $s^*(i) \in B_{i \oplus 1}$,
by the definition of $B_{i \oplus 1}$ and the fact that $s^*$ is a
Nash equilibrium, we have that $s^*(i \oplus 1) = s^*(i)$.
This is because $s^*(i)$ strictly dominates any other strategy
choice for node $i \oplus 1$.
By line 6 in
Algorithm \ref{alg:cycle-ene}, we have $X_{i \oplus 1} = X_i$ and therefore, $s^*(i
  \oplus 1) \in X_{i \oplus 1}$.

If $X_i \not\subseteq B_{i \oplus 1}$ and $s^*(i \oplus 1) \in
  A_{i \oplus 1}$ then by line 4 of Algorithm \ref{alg:cycle-ene}, we
  have $s^*(i \oplus 1)
\in X_{i \oplus 1}$. If $s^*(i \oplus 1) \not\in A_{i \oplus 1}$,
then since $s^*$ is a Nash equilibrium, $s^*(i \oplus 1) = s^*(i)$ and
$s^*(i \oplus 1) \in C_{i \oplus 1}$ (otherwise node $i \oplus 1$ would have a
profitable deviation to a strategy in $A_{i \oplus 1}$). Therefore, by
line 4 of Algorithm \ref{alg:cycle-ene}, we have $s^*(i \oplus 1) \in
X_i \cap C_{i\oplus 1} \subseteq X_{i \oplus 1}$.

Conversely, suppose the output of Algorithm \ref{alg:cycle-ene} is
YES. From the definition, this implies that for all $i \in V$,
$X_i \neq \emptyset$ and for all $j \in Q$: $q(j) \in X_j$ (in fact,
$X_j =\{q(j)\}$). We define a Nash equilibrium $s^*$ as
follows. First, let $s^*(0)=q(0)$.  Next we assign values to $s^*(i)$
starting at $i = n-1$ and going down to $i=1$ as described below.
\begin{itemize}
\itemsep0em 
\item If $i \in Q$ then $s^*(i)=q(i)$.
\item If $i \not\in Q$ and $X_{i} \subseteq B_{i\oplus 1}$
  then by Algorithm \ref{alg:cycle-ene} we have $X_{i \oplus 1} = X_{i}$. Let $s^*(i) = s^*(i \oplus 1)$.
\item Assume $i \not\in Q$ and $X_{i} \not\subseteq
  B_{i \oplus 1}$. If $s^*(i\oplus 1) \in X_{i} \cap C_{i \oplus 1}$ set
  $s^*(i)=s^*(i\oplus 1)$. Otherwise $s^*(i\oplus 1) \in A_{i\oplus 1}$ and we set
  $s^*(i)$ to any element in $X_{i} \setminus B_{i\oplus 1}$. 
\end{itemize}

It is straightforward to verify that for the joint strategy $s^*$
defined as above, for all $i \in V$, $s^*(i) \in X_i$. We now argue
that $s^*$ is a Nash equilibrium. Suppose not, then there exists $j
\in V$ and a strategy $x \in C(j)$ such that $p_j(x, s^*_{-j}) >
p_j(s^*)$.
 We have the following cases.

\smallskip
\noindent{\em Case $j \not\in Q$.}
If $X_{j \ominus 1} \subseteq B_{j}$ then by the
  definition of $s^*$, we have $s^*(j) = s^*(j\ominus 1)$ and so $x \neq s^*(j\ominus 1)$. By the
  definition of $B_{j}$ and $X_{j \ominus 1}$, we have that for all strategies
  $y \in C(j)$: $\beta(j, s^*(j)) + \wgtone{j} - 1 \geq \beta(j, y)$. 
  Now, we have that $p_{j}(s^*) = \beta(j, s^*(j)) + \wgtone{j} \geq 
  \beta(j, x) + 1 > p_j(x, s^*_{-j})$ which is a contradiction.

If $X_{j \ominus 1} \not\subseteq B_j$ and $s^*(j) \in X_{j \ominus 1}
\cap C_j$ then by the definition of $s^*$ we have $s^*(j) =
s^*(j \ominus 1)$, and so $x \neq s^*(j\ominus1)$. By the definition of $C_j$,
we have that $p_j(s^*) = \beta(j,s^*(j)) + \wgtone{j}
\geq \beta(j, x) = p_j(x, s^*_{-j})$; a contradiction.

If $X_{j \ominus 1} \not\subseteq B_j$ and $s^*(j) \not\in X_{j \ominus
  1} \cap C_j$ then by the definition of $s^*$, we have $s^*(j) \in A_j$ and
$s^*(j \ominus 1) \not\in B_{j}$. %
From the former, $\beta(j,s^*(j)) \geq \beta(j,y)$ for all strategies $y$.
From the latter, it follows that $\beta(j, x) \leq \beta(j, s^*(j)) + \wgtone{j}$, 
because all bonuses are integers. 
Thus $p_j(s^*) = \beta(j, s^*(j)) \geq p_j(y, s^*_{-j})$ for all $y \in C(j)$; a contradiction.

\smallskip
\noindent{\em Case $j \in Q$.}
Consider the value of $X_j$ in line 7 of
  Algorithm \ref{alg:cycle-ene} during the iteration when $i =
  j \ominus 1$. Since the output of Algorithm \ref{alg:cycle-ene} is
  assumed to be YES, we have that $q(j) \in X_j$. Now applying a
  similar case analysis as above, we can argue that node $j$ does not
  have a profitable deviation from $s^*$. 
\qed
\end{proof}

\anesimplecycles*
\begin{proof}
We first show that in the special case of unweighted simple cycles, for which \NE always exists, 
Algorithm \ref{alg:cycle-ane} solves the \ane problem for
unweighted simple cycles in $\mathcal{O}(|G|)$ time.
In other words, we argue that given an unweighted simple cycle over
the nodes $V=\{0,\ldots,n-1\}$ and a query $q: Q \to C$, the output of
Algorithm \ref{alg:cycle-ane} is YES iff for every Nash equilibrium $s^*$, the joint
strategy $s^*$ is consistent with $q$.

Suppose there exists a Nash equilibrium $s^*$ which is not consistent
with $q$. For the sake of simplicity, assume that $s^*(0)=q(0)$ and let
$j$ be the minimal index such that $s^*(j) \neq q(j)$. By induction we
argue that if Algorithm \ref{alg:cycle-ane} does not terminate with an output NO before
the iteration with $i= j \ominus 1$, then for all $k$ such that $0 \leq k \leq
j$: $s^*(k) \in X_k$. Since $s^*(0) = q(0)$, line 1 of Algorithm
\ref{alg:cycle-ane} implies $X_0 = \{s^*(0)\}$. Assume that we have $s^*(i) \in X_i$ (for $i <
j$) and consider the iteration of the loop in line 2 of Algorithm \ref{alg:cycle-ane}
for $i \oplus 1$.  We have the following cases.

If $X_i \subseteq B_{i \oplus 1}$ then $s^*(i) \in B_{i \oplus 1}$, by
the definition of $B_{i \oplus 1}$ and the fact that $s^*$ is a Nash
equilibrium, we have that $s^*(i \oplus 1) = s^*(i)$. 
This is because $s^*(i)$ strictly dominates any other strategy choice for node $i \oplus 1$.
By line 6 in
Algorithm \ref{alg:cycle-ane}, we have $X_{i \oplus 1} = X_i$ and therefore, $s^*({i
  \oplus 1}) \in X_{i \oplus 1}$.

If $X_i \not\subseteq B_{i \oplus 1}$ and $s^*(i \oplus 1) \in A_{i
  \oplus 1}$ then by line 4 of Algorithm \ref{alg:cycle-ane}, we have $s^*(i \oplus 1)
\in X_{i \oplus 1}$. If $s^*(i \oplus 1) \not\in A_{i \oplus 1}$,
then since $s^*$ is a Nash equilibrium, $s^*(i \oplus 1) = s^*(i)$ and
$s^*(i \oplus 1) \in C_{i \oplus 1}$ (otherwise node $i \oplus 1$ has a
profitable deviation to a strategy in $A_{i \oplus 1}$). Therefore, by
line 4 of Algorithm \ref{alg:cycle-ane}, we have $s^*(i \oplus 1) \in
X_i \cap C_{i\oplus 1} \subseteq X_{i \oplus 1}$.

Now consider the iteration of Algorithm \ref{alg:cycle-ane} when $i =
j \ominus 1$. By the above argument, $s^*(j) \in X_j$ and by
assumption $s^*(j) \neq q(j)$. Therefore, the condition on line 8 is
satisfied and the output of the algorithm is NO.

Conversely, suppose the output of Algorithm \ref{alg:cycle-ane} is NO. Let $j$ be the
index such that the algorithm terminates with $i=j$. This implies that
$j \in Q$ and $\{q(j)\} \neq X_j$. Note that by definition, $X_j \neq
\emptyset$. Define a partial joint strategy $s_1$ on the nodes
$\{0,\ldots,j\}$ inductively as follows. Let $s_1(0)=q(0)$ and $s_1(j)
\in X_j \setminus \{q(j)\}$. We define $s_1(i)$ starting at $i = j - 1$
going down to $i = 1$ as follows.
\begin{itemize}
\item If $i \in Q$ then $s_1(i)=q(i)$.
\item If $i \not\in Q$ and $X_{i} \subseteq B_{i \oplus 1}$
  then by Algorithm \ref{alg:cycle-ane} we have $X_{i} = X_{i \oplus 1}$. Let $s_1(i) = s_1(i \oplus 1)$.
\item Assume $i \not\in Q$ and $X_{i} \not\subseteq
  B_{i \oplus 1}$. If $s_1(i\oplus 1) \in X_{i} \cap C_{i \oplus 1}$ we set
  $s_1(i)=s_1(i\oplus 1)$. Otherwise we have that $s_1(i\oplus 1) \in A_{i\oplus 1}$ and we set
  $s_1(i)$ to any element in $X_{i} \setminus B_{i\oplus 1}$.
\end{itemize}

We can then extend $s_1$ to a joint strategy $s_2$ by allowing nodes
$j+1, j+2, \ldots n-1$ to switch, in this order, to their best response strategies.
Note that this is well defined since the best response
of a node $i$ depends only on the strategy of its unique predecessor
$i \ominus 1$ on the cycle. If $s_2$ is a Nash equilibrium then we
have a joint strategy which is not consistent with $q$. If $s_2$ is
not a Nash equilibrium then we can argue that node 0 is not playing
its best response in $s_2$. Let node 0 switch to its best response,
denoted by $x$. By definition of $s_1$: $x \neq q(0)$. Now by applying the
best response improvement to each node successively in the order
$1,2,\ldots, n-1$ we can show it is possible to construct a joint
strategy $s_3$ which is a Nash equilibrium in which
$s_3(0)=x$. Details of this construction can be found
in \cite[Lemma 6]{ASW15}. Thus it follows that $s_3$ is a Nash
equilibrium which is not consistent with $q$. 

Finally, we show that Algorithm \ref{alg:cycle-ane-weighted} solves the \ane problem in $\calO(m|G|)$ time for
weighted simple cycles.
First, suppose there exists a Nash equilibrium $s^*$ which is not consistent
with a \ane query $q : Q \to M$. Consider the iteration of the main loop of Algorithm \ref{alg:cycle-ane-weighted} for $c = s^*(0)$.
Note that Algorithm \ref{alg:cycle-ene} for $q'(0) = s^*(0)$ would return YES, because $s^*$ is consistent with $q'$.
From the proof of Theorem \ref{thm:ene-simple-cycles} we know that, for every $i\in V$, the set $X_{i}$ this algorithm computes is equal to the set of colours node $i$ can have in any \NE consistent with $q'$. Note that there has to be $i^* \in Q$ such that $q(i^*) \neq s^*(i^*) \in X(i^*)$, because $s^*$ is not consistent with $q$. Thus Algorithm \ref{alg:cycle-ane-weighted} returns NO, because $X_{i^*} \neq \{q(i^*)\}$.

Conversely, suppose Algorithm \ref{alg:cycle-ane-weighted} returns NO for a query $q$. Then, there exists $i^* \in Q$ for which $X_{i^*}$ computed by Algorithm \ref{alg:cycle-ane-weighted} is $\neq \{q(i^*)\}$. Let us pick any $x \in X_{i^*} \setminus \{q(i^*)\}$.
Based on the interpretation of the set $X(i^*)$, there exists a \NE $s^*$ such that $s^*(i^*) = x$.
Such $s^*$ would not be consistent with $q$, which concludes the proof.
\qed
\end{proof}

\acyc*
\begin{proof}
Suppose there is such a sequence. 
From the definition of $\succ_s$ there exist players $i_1,\ldots,i_k$ 
such that $\{x_j,x_{j+1}\} \subseteq C(i_j)$ and $s(i_j) = x_j$ for all $j=1,\ldots,k$ 
(where we identify $x_{k+1}$ with $x_1$).
For a  joint strategy $s$ and colour $c \in M$, 
let $\#c(s)$ denote the number of players who chose colour $c$ in $s$,
i.e. $\#c(s) = |\{v \in V | s(v) = c\}|$.
Note that for all $j$ player $i_j$'s payoff in $s$ is
$\#x_j(s) - 1$ and
switching to $x_{j+1}$ would give him payoff $\#x_{j+1}(s)$.
Therefore, 
$\#x_j(s) - 1 \geq \#x_{j+1}(s)$, because otherwise $s$
would not be a \NE.
However, this implies $\#x_1(s) - k \geq \#x_1(s)$; a contradiction. \qed
\end{proof}

\acycorder*
\begin{proof}
Let $\succ$ be an acyclic relation on a finite set $S$ and $k = |S|$. 
The directed graph defined by $G = (S,\succ\nobreak)$ is a DAG, because
$\succ$ is acyclic.  Therefore we can topologically sort all the
elements in $S$ into a sequence $x_1,\ldots,x_k$ in such a way that
$x_i \succ x_j$ implies $i \leq j$. Notice that a relation $\succ^*$
defined as $x_i \succ^* x_j$ iff $i \leq j$ is a total order on
$S$.\qed
\end{proof}

\lemordertwone*
\begin{proof}
Suppose that $\SP(\succeq^*_s)(i) \neq s(i)$ for some player $i$.
This means $s(i) \neq \max_{\succeq^*_s} C(i)$, so there exists $x \in C(i)$ such
that $x \succeq^*_s s(i)$ and $x \neq s(i)$. However, $\{x,s(i)\} \subseteq C(i)$
implies that $s(i) \succ x$ and so also $s(i) \succeq^*_s x$ should hold; 
a contradiction with the fact that $\succeq^*_s$, as a total order, is antisymmetric.
\qed
\end{proof}

\begin{example}
\label{ex:many-NEs}
Let the set of colours $M$ be $\{1,\ldots,m\}$ and consider a clique
consisting of $(m-1)m/2$ players.
For every $x,y \in M$ such that $x < y$ there is exactly one player in this clique whose available
colours are $x$ and $y$ only. It is easy to see that for the total order $\succeq$
defined as $m \succeq m-1 \succeq \ldots \succeq 1$ the number of
players choosing colour $m$ in $\SP(\succeq)$ is $m-1$, which is the
maximum possible.  It can be verified that in $\SP(\succeq)$, all the
players who picked colour $x$ receive a payoff of $x-2$, each colour
gives a different payoff and no player can improve his payoff. It
follows that $\SP(\succeq)$ is a Nash equilibrium. 
If we consider any other total order on
$M$, it will result in a permutation of this sequence of payoffs.
Because all of these numbers are different, no
two joint strategies induced by two different total orders are the
same. %
\end{example}

\enedagonelemma*
\begin{proof}
Let $\theta=(i_1,\ldots,i_n)$ be the topological ordering of $V$
chosen in line 1 of Algorithm \ref{alg:out1e}. We show that for all
$j: 1 \leq j \leq n$, for all $c \in X(i_j)$, there exists a Nash
equilibrium $s^*$ such that $s^*(i_j)=c$ and for all $k \neq j$,
$s^*(i_k) \in X(i_k)$. For $i \in V$, let $A_i =\{c \in C(i) \mid \beta(i,c) \geq
\beta(i,c')$ for all $c' \in C(i)\}$ be the set of colours available to player $i$
with the maximum bonus.

Let $i_j$ be a node such that $N_{i_j}= \emptyset$. In the iteration
of the algorithm which considers node $i_j$, we have $Y'=\emptyset$ 
every colour which does not belong to $A(i_{j})$ is removed in line 10.
Thus
$X(i_j)=A(i_j)$.  Let $D=\{i_k \in V \mid N_{i_k} = \emptyset\}$.
Consider the partial joint strategy $s':D \to C$ defined as
$s'({i_j})=c$ and for $i_k \in D$ such that $i_k \neq i_j$ let
$s'(i_k) \in A(i_k)$. Now $s'$ can be extended to a joint strategy
$s^*:V\to C$ by successively making each node (according to the
ordering $\theta$) choose its best response. Since $G$ is a DAG, it
easily follows that $s^*$ is a Nash equilibrium and 
for all $i_k$, $s^*(i_k) \in X(i_k)$.

Now consider a node $i_m$ such that $N_{i_m} \neq \emptyset$ and let
$c \in X(i_m)$. Let $D=\{i_m\} \cup \{i_j \in V \mid \text{ there is a
path from $i_j$ to $i_m$ in $G$}\}$. By definition, for all $i_j$ in
$D$, we have $j \leq m$ (according to the ordering $\theta$). Consider
the partial joint strategy $s': D \to C$ defined inductively as
follows. Let $s'(i_m)=c$. Suppose that $s'(i_j)$ is already defined for some $i_j \in D$,
then for each $i_k \in N_{i_j}$ we do the following.
If $s'(i_j) \in X(i_k)$ then let
$s'(i_k)=s'(i_j)$. If $s'(i_j) \not\in X(i_k)$, then consider the
iteration of Algorithm \ref{alg:out1e}, which adds $s'(i_j)$ to
$X(i_j)$. If $X(i_k)$ is removed from $Y'$ in line 10 because of
colour $c'$ then let $s'(i_k) = c'$.
Otherwise, if the corresponding maximum bipartite matching in line 12
matches $X(i_k)$ with $(c',x)$, then define $s'(i_k)=c'$.
Since the out-degree of $G$ is at most 1, $s'(i_k)$ is assigned a value exactly once
and so $s'$ is a valid function.

By definition of $s'$, for all $i_j \in D$, $s'(i_j) \in
X(i_j)$. Given a node $j \in D$, a partial joint strategy $s: D \to
C$ and $c \in C(j)$, let $N_{j}(s,c)=\{k \in N_{j} \mid s(k)=c\}$.
We now argue that for all $i_j \in D$, $s'(i_j)$ is a best response
for node $i_j$ to $s'_{-i_j}$. 

Suppose $N_{i_j}=\emptyset$. Since
$s'(i_j) \in A(i_j)$, it follows that $s'(i_j)$ is a best response to
$s'_{-i_j}$. 
Now suppose $N_{i_j} \neq \emptyset$ and $s'(i_j)$ is not a best response
to $s'_{-i_j}$. Then there exists a $c' \in C(i_j)$ such that
$p_{i_j}(c',s'_{-i_j}) > p_{i_j}(s')$. This implies that
$|N_{i_j}(s',c')| + \beta(i_j, c') > |N_{i_j}(s', s'(i_j))| + \beta(i_j, s'(i_j))$.
Consider the iteration
of Algorithm \ref{alg:out1e} in which $s'(i_j)$ is added to
$X(i_j)$. By the definition of $s'$ and Algorithm \ref{alg:out1e}, in
this iteration, $|S| = |N_{i_j}(s', s'(i_j))|$. 
If $c'$ is removed from $C'$ in line 10 then
we would have $p_{i_j}(s') = |N_{i_j}(s', s'(i_j))| + \beta(i_j, s'(i_j)) =
|S| + \beta(i_j, s'(i_j)) \geq |Y'| + \beta(i_j, c') \geq |N_{i_j}(s',c')| + \beta(i_j, c') = p_{i_j}(s')$; a contradiction. 
Therefore for the bipartite graph $G'$ constructed in line 11 we need to have $N_{i_j}(s',c') \in Y'$.
Notice that every node in $Y'$ is matched with some other node in line 12, 
because the size of the matching is $|Y'|$.
Again, by the definition of $s'$, for all nodes $i_k \in N_{i_j}(s',c')$
the node %
$X(i_k)$ would need to be matched with $(c',x)$ for some $1 \leq x \leq |S| + \beta(i_j, s'(i_j)) - \beta(i_j, c')$. But this is impossible, because
$|N_{i_j}(s',c')| > |S| + \beta(i_j, s'(i_j)) - \beta(i_j, c')$,
thereby contradicting the assumption that
$s'(i_j) \in X(i_j)$.

As in the earlier case, $s'$ can now be extended to a joint strategy
$s^*:V \to C$. First, for all $j \in \{i \mid N_i
= \emptyset\} \setminus D$, set $s'(j)\in A(j)$. Then successively
make each node according to the ordering $\theta$ choose its best
response. The resulting joint strategy $s^*$ is a Nash equilibrium.
\qed
\end{proof}

\enedagone*
\begin{proof}
We show that given an unweighted DAG $G=(V,E)$ with out-degree at most 1 and a
query $q$, Algorithm \ref{alg:out1e} returns YES iff there
exists a Nash equilibrium $s^*$ which is consistent with $q$.

Suppose Algorithm \ref{alg:out1e} returns YES. Then from the
definition, for all $i \in V$, $X_i \neq \emptyset$ and for all $j \in
P$, $X_j=\{q(j)\}$. By Lemma~\ref{lm:DAG-out1-ene} it follows that
there exists a Nash equilibrium $s^*$ which is consistent with $q$.

Conversely, suppose there exists a Nash equilibrium $s^*$ which is
consistent with $q$. Let $\theta=(i_1,\ldots,i_n)$ be the topological
ordering of $V$ chosen in line 1 of Algorithm \ref{alg:out1e}. We
argue that for all $j \in \{1,\ldots,n\}$, $s^*(i_j) \in X(i_j)$. The
claim follows easily for $i_1$. Consider a node $i_m$ and suppose for
all $j < m$, $s^*(i_j) \in X(i_j)$. For $c \in C$, let $N_{i_m}(s^*,c)
=\{i_k \in N_{i_m} \mid s^*(i_k) =c\}$. Since $s^*$ is a Nash
equilibrium, $s^*(i_m)$ is a best reponse to the choices made by all
nodes $i_k \in N_{i_m}$. This implies that for all $c \neq s^*_{i_m}$,
$|N_{i_m}(s^*,c)| + \beta(i_{j}, c) \leq |N_{i_m}(s^*,s^*_{i_m})| + \beta(i_{j}, s^*_{i_m})$.
This condition guarantees that Algorithm \ref{alg:out1e} will
find a matching of size $|Y'|$ for $G'$ defined in line 11 and thus
$s^*(i_m) \in X(i_m)$.

The computational complexity of Algorithm \ref{alg:out1e}
mainly depends on the maximum matching algorithm in bipartite graphs
used at line 11.
There are several such algorithms, each with a different computational complexity.
We can use the standard Hopcroft-Karp algorithm which has complexity $\calO(E\sqrt{V})$ where
$E$ is the number of edges and $V$ is the number of nodes in a given bipartite graph.
For $l = 1,\ldots,n$, let $Y_l$ denote the value of $Y$ at line 4 of Algorithm \ref{alg:out1e} for $j := l$. 
Note that all of these sets are disjoint, because each node is a predecessor of
at most one other node.
Let $f(x)$ be the function that returns the maximum running time of one iteration of the loop
between lines 6-13 for a set $Y$ of size $x$.
Note that for any $j$, this loop is executed at most once for each colour.
Consider the function $g(x) := f(x) - f(0)$; note that $g(x) = \calO(f(x))$.
It is easy to see that $g(x)$ is increasing, convex (the complexity of the matching problem is at least linear in $x$), 
and $g(0)$ = 0.
We will show that such defined $g$ is superadditive, i.e.~$g(a) + g(b) \leq g(a+b)$ for any $a,b \geq 0$.
\begin{lemma}
Any convex, increasing function $h$ such that $h(0) = 0$ is superadditive.
\end{lemma}
\begin{proof}
Consider any $a,b \geq 0$ and the linear function $q(x) := x \cdot h(a+b)/(a+b)$,
which is the line connecting the $(0,0)$ and $(a+b,h(a+b))$ points on the curve defined by function $h$.
As $h$ is convex we have that any point along $q(x)$ for $x \in [0,a+b]$
is at least as high as $h(x)$. In particular, $q(a) \geq h(a)$ and $q(b) \geq h(b)$.
At the same time $h(a+b) = q(a) + q(b)$, which concludes the proof. \qed
\end{proof}

We now have that the total running time of Algorithm \ref{alg:out1e} is
$\calO(m \cdot \sum_{j=1}^{n} f(|Y_j|)) = 
\calO(m\cdot \sum_{j=1}^{n} (g(|Y_j|) + f(0))) =
\calO(m\cdot \sum_{j=1}^{n} g(|Y_j|)) = 
\calO(m\cdot g(\sum_{j=1}^{n} |Y_j|)) 
= \calO(m \cdot f(|V|))$, where the third equality holds because $g$ is superadditive.
As a result, it suffices to estimate the matching time for a graph $G'$ with 
$|Y'| = |V|$.
Due to lines 9-10 we have $|S| + \beta(i_j, c) - \beta(i_j, c') \leq |Y'|$ $=$ $\calO(n)$, 
For such a $Y'$, the bipartite graph $G'$ at line 11 would have $n + nm = \calO(nm)$ nodes and $n \cdot nm = n^2m$ edges. The complexity of one matching for such a $G'$ is $\calO(n^{2}m\sqrt{nm})$.
This implies that the total running time of Algorithm \ref{alg:out1e} is $\calO((nm)^{5/2}) = \calO(|G|^{5/2})$, because $|G| = \calO(nm)$.
\qed
\end{proof}

\anedagone*
\begin{proof}[sketch]
The proof is similar to that of Theorem~\ref{thm:ene-dag-out1}. Let
$\theta=(i_1,\ldots,i_n)$ be the topological ordering of $V$ chosen in
line 1 of Algorithm~\ref{alg:out1a}. First, we need the following lemma
which essentially follows from the proof of Lemma \ref{lm:DAG-out1-ene}.

\begin{lemma}
\label{lm:DAG-out1-ane}
For $j \in \{1,\ldots, n\}$, in the iteration of
Algorithm~\ref{alg:out1a} for the node $i_j$, consider the set
$X(i_j)$ computed in line 3 (lines 3-13 in the long version).
For all $c \in X(i_j)$, there exists a
Nash equilibrium $s^*$ such that $s^*(i_j)=c$ and for all $k \leq j$,
$s^*(i_k) \in X(i_k)$.
\end{lemma}

Suppose the output of Algorithm~\ref{alg:out1a} is NO. Let $i_j$ be
the node which is being processed when the algorithm outputs NO. This
implies that $i_j \in Q$ and there exists a $c \neq q(i_j)$ such that
$c \in X(i_j)$. From Lemma \ref{lm:DAG-out1-ane} there exists a Nash equilibrium $s^*$
such that $s^*(i_j)=c$. Thus $s^*$ is a Nash equilibrium which is not
consistent with $q$.

Suppose there is a Nash equilibrium $s^*$ which is not consistent with
$q$. Let $i_j$ be the first node in the ordering $\theta$ such that
$s^*(i_j) \neq q(i_j)$. We can argue that for all $k <j$,
$s^*(i_k) \in X(i_k)$ and in the iteration of
Algorithm~\ref{alg:out1a} for the node $i_j$, at line 3, $s^*(i_j) \in
X(i_j)$. This implies that the condition on line 4 is satisfied and
the algorithm outputs NO. The bound on the running time follows
from the analysis given in the proof of
Theorem~\ref{thm:ene-dag-out1}. \qed
\end{proof}

\end{document}